\documentclass{amsart}
\usepackage{amsmath}
\usepackage{amssymb}
\usepackage{amsfonts}
\usepackage{amsthm}
\usepackage{hyperref}
\usepackage{latexsym}
\usepackage{graphics}
\usepackage{graphicx}
\usepackage[all]{xy}
\usepackage[italian,english]{babel}
\newtheorem{Lem}{Lemma}
\newtheorem{Def}{Definition}
\newtheorem{Prop}{Proposition}

\newtheorem{Theorem}{Theorem}

 \usepackage{epsfig}
 
\DeclareMathOperator{\im}{Im}
\DeclareMathOperator{\re}{Re}

\begin{document}


\title{Poisson Sigma Model with  Branes and Hyperelliptic Riemann Surfaces}

\author{Andrea Ferrario}
\address{Departement Mathematik--ETH Z\"urich--8092, Z\"urich---Switzerland}  
\email{andrea.ferrario@math.ethz.ch}

\maketitle

\begin{abstract}
We derive the explicit form of  the superpropagators  in presence of general boundary conditions (coisotropic branes) for the Poisson Sigma Model. This  generalizes the results presented in \cite{Cattaneo:1999fm} and \cite{CF} for  Kontsevich's  angle function \cite{Kont} used in the deformation quantization program of Poisson manifolds. The relevant  superpropagators for $n$ branes are defined as gauge fixed homotopy operators of a complex of differential forms on  $n$ sided polygons $P_n$ with particular "alternating" boundary conditions.
 In presence of more than three branes we use first order Riemann theta functions with odd singular characteristics on the Jacobian variety of a hyperelliptic Riemann surface (canonical setting). In genus $g$ the  superpropagators present  $g$ zero modes contributions. 
\end{abstract}

\tableofcontents

\section{Introduction}

The Poisson Sigma Model (P$\sigma$M) is a topological field theory defined in terms of a functional on the space of maps from the tangent bundle of a two dimensional oriented surface $\Sigma$ to the cotangent bundle of a given Poisson manifold $(M,\pi)$. When the source surface is the unit disk it has been shown in the celebrated paper \cite{Cattaneo:1999fm} that  Kontsevich's star product on $M$ \cite{Kont} can be obtained from  Feynman's expansions of certain Green functions, assuming particularly simple boundary conditions. The same authors have extended the previous calculations with the P$\sigma$M with boundary conditions such that the base map $X$ maps the boundary of the unit disk to a coisotropic submanifold of $M$ \cite{CF}: it turns out that the coisotropic submanifolds of a Poisson manifold label the possible boundary conditions of the P$\sigma$M and its quantization is related to the deformation quantization of the coisotropic submanifold itself. 
In \cite{Calvo:2005th} it has been shown that even non coisotropic branes are allowed at quantum level; when the brane is defined by the so called second class constraint, then the perturbative quantization of the P$\sigma$M yields  Kontsevich's star product associated now to the Dirac bracket defined on the brane. In \cite{CZ,C2007} a unifying approach is proposed with the introduction of   Pre-Poisson submanifolds: given a Pre-Poisson submanifold $C$ of a Poisson manifold $M$, then it is always possible to find a  presymplectic submanifold $M'$ of $M$ containing $C$ as coisotropic submanifold.
In this note we construct explicitly the superpropagators for  the P$\sigma$M in presence of $n \geq 2$ coisotropic branes generalizing the results in \cite{Cattaneo:1999fm,CF}. 

In the perturbative expansion every coisotropic submanifold $\mathcal{C}_j$, or brane, of the Poisson manifold $(M,\pi)$ is defined by the constraint
$\mathcal{C}_j=\left\{ x^{\mu_j}=0 \mid \mu_j \in I_j  \right\}$; considering $n \geq 2 $ branes the source manifold   for the P$\sigma$M is defined by the couple $(P_n, u)$ where $P_n := u(\mathbb{H}^+)$ is a $n$ sided  polygon  and $u :\mathbb{H}^+\rightarrow P_n$ is a suitable homeomorphism between the compactified complex upper half plane $\mathbb{H}^+$ and $P_n$, depending on the number  of branes considered. This definition allows to fix the polygon $P_n$ through $u$ and to use the technique of the "mirror charges" \cite{CF} to write the explicit formula for the superpropagators with the correct boundary conditions. In presence of $n=2,3$ branes, $u$ is chosen to be a suitable Schwarz-Christoffel mapping \cite{Gonzal}, while for $ n\geq 4$ branes we need to introduce particular sections of hyperelliptic Riemann surfaces and their projections to the Riemann sphere $\mathbb{C}_{\infty}$ (we refer to Section 6 for the full construction). In the sequel each side $\partial P_n^i$ of $P_n$ is called brane as well, when confusion does not arise with the corresponding coisotropic submanifold $\mathcal{C}_i \subset M$. 
 The fundamental superpropagators in presence of $n$ branes  are those associated to the boundary conditions expressed by the index sets $\mathcal{S}_1=I_1^c\cap I_2 \cap I^c_3\cap \dots \cap I_n$,  $\mathcal{S}_2=I_1\cap I_2^c \cap I_3\cap \dots \cap I_n^c$ for $n$ even and $\mathcal{S}_1=I_1^c\cap I_2 \cap I^c_3\cap \dots \cap I^c_n$,  $\mathcal{S}_2=I_1\cap I_2^c \cap I_3\cap \dots \cap I_n$ for $n$ odd; we call them relevant. Let $(\mathcal{H}^n_i,d)$ be the  complexes of differential forms on  $P_n$ with Dirichlet boundary conditions on the even branes  for $i=1$ and  Dirichlet boundary conditions on the odd ones for $i=2$. The relevant superpropagators  are defined as the gauge fixed homotopy operators $G_{\mathcal{S}_i}$  which give the Hodge-Kodaira decomposition

\begin{equation*}
d G_{\mathcal{S}_i}+G_{\mathcal{S}_i} d = I-P_{\mathcal{S}_i}
\end{equation*}
of $\mathcal{H}_i^n$, where $P_{\mathcal{S}_i}$ is a projection  onto De Rham cohomology, for $i=1,2$. This setup is a special case of a strong deformation retract used in the homological perturbation theory contest (\cite{Kajiura:2001ng,Kajiura:2003ax} and references therein).  The degree minus one operators $G_{\mathcal{S}_i}:\mathcal{H}_i^{n,k} \rightarrow \mathcal{H}_i^{n,k-1}$ are given in Def.1; the paper is devoted to the construction of their integral kernels $\theta(Q,P)_{\mathcal{S}_i}:=-\frac{i}{\hbar}\langle \tilde{\xi}^{a_1}(Q)\tilde{\eta}_{a_2}(P)\rangle$, $(P,Q) \in P_n \times P_n$, $a_1,a_2 \in \mathcal{S}_i$ satisfying the  boundary conditions imposed by the $\mathcal{S}_i$ themselves. With $\tilde{\xi}, \tilde{\eta}$ we denote superfields \cite{Cattaneo:1999fm} of the P$\sigma$M  in the perturbative expansion. In principle $\theta(Q,P)_{\mathcal{S}_i}$ will consist of two different contributions: a generalization of the Kontsevich angle map \cite{Kont} and an additional term   due to the projection $P_{\mathcal{S}_i}$ onto cohomology $P_{\mathcal{S}_i}\mathcal{H}^n_i:=H^{\bullet}(\mathcal{H}^n_i)$.
 The main results of the paper are collected in the following

\begin{Theorem}{\bf{Superpropagators for the P$\sigma$M with $n$ branes}} 
Let $G_{\mathcal{S}_i}$ be the  relevant superpropagators  for the P$\sigma$M with $n$ branes defined by  the constraints \ $\mathcal{C}_j=\left\{ x^{\mu_j}=0 \mid \mu_j \in I_j  \right\}$ and $\mathcal{S}_i$ defined as above. The integral kernels $\theta(Q,P)_{\mathcal{S}_i}:=-\frac{i}{\hbar}\langle \tilde{\xi}^{\bullet}(Q)\tilde{\eta}_{\bullet}(P)\rangle$  are given by:

\begin{itemize} 
\item  two branes case: 
\begin{eqnarray*}
&&\theta(Q,P)_{\mathcal{S}_1} = \frac{1}{2\pi}d \arg \frac{(u-v)(\bar{u}-v)}{(\bar{u}+v)(u+v)}, \\
&&\theta(Q,P)_{\mathcal{S}_2} = \frac{1}{2\pi}d \arg \frac{(u-v)(\bar{u}+v)}{(\bar{u}-v)(u+v)},
\end{eqnarray*}
where $P_2 :=u(\mathbb{H}^+)$ with $u(z)=\sqrt{z}$, $ v:=u(w)$, $d=d_u+d_v$: we identify $(P,Q)$ with the couple $(u,v)$.

\item  three branes case: 
\begin{eqnarray*}
&&\theta(Q,P)_{\mathcal{S}_1} = \frac{1}{2\pi}d\arg \frac{\sin i\pi (u-v)\sin i\pi (\bar{u}+v)}{\sin i\pi (\bar{u}-v)\sin i\pi (u+v)}, \\
&&\theta(Q,P)_{\mathcal{S}_2} =\frac{1}{2\pi}d\arg \frac{\sin i\pi (u-v)\sin i\pi (\bar{u}-v)}{\sin i\pi (\bar{u}+v)\sin i\pi (u+v)},
\end{eqnarray*}
where $P_3 :=u(\mathbb{H}^+)$ with $u(z)= \frac{1}{2\pi} \int_1^z \frac{ds}{\sqrt{s(s-1)}}$, $ v:=u(w)$, $d=d_u+d_v$: we identify $(P,Q)$ with the couple $(u,v)$.

\item $n=2g+2$, $g \geq 1$  branes case:

 \begin{eqnarray*}
\theta(Q,P)_{\mathcal{S}_1}=\frac{1}{2\pi}d\arg \frac{ \vartheta(\varphi(P)-\varphi(Q)+\mathcal{A}_g,\Omega) \vartheta(\overline{\varphi(P)}-\varphi(Q)+ \bar{\mathcal{A}}_g,\Omega)}{\vartheta(\varphi(P)+\varphi(Q)+\bar{\mathcal{A}}_g,\Omega) \vartheta(\overline{\varphi(P)}+\varphi(Q)+\mathcal{A}_g,\Omega)}-4\mathcal{Z}_{\mathcal{S}_1}(Q,P) & \\
\theta(Q,P)_{\mathcal{S}_2}=\frac{1}{2\pi} d\arg\frac{ \vartheta(\varphi(P)-\varphi(Q)+\mathcal{A}_g,\Omega) \vartheta(\overline{\varphi(P)}+\varphi(Q)+ \bar{\mathcal{A}}_g,\Omega)}{\vartheta(\varphi(P)+\varphi(Q)+\bar{\mathcal{A}}_g,\Omega) \vartheta(\overline{\varphi(P)}-\varphi(Q)+\mathcal{A}_g,\Omega)}-4\mathcal{Z}_{\mathcal{S}_2}(Q,P)& 
 \end{eqnarray*}
\end{itemize}
 where $\mathcal{Z}_{\mathcal{S}_i}(Q,P):=\left\{ \begin{array}{c} \im \varphi_i(Q) (\im \Omega)^{-1}d\re  \varphi_j(P)~~~ i=1 \\ \im \varphi_i(P) (\im \Omega)^{-1}d\re  \varphi_j(Q)~~~ i=2\end{array}\right. $,  $\varphi: \mathcal{M} \rightarrow \mathcal{J}(\mathcal{M}):=\mathbb{C}^g/\null^tnI+\null^tm\Omega$ is the Abel-Jacobi map for the hyperelliptic Riemann surface $\mathcal{M}$ of genus $g$ which realizes the two sheeted branched covering $z: \mathcal{M} \rightarrow \mathbb{C}_{\infty}$ with 
 branching points $\{ P_i\}_{i=1,\dots,2g+2}$ such that $z(P_i) \in \mathbb{R} \cup \{\infty\}$,  $(P,Q) \in P_n \times P_n$ and $d=d_P+d_Q$.
 The  polygon $P_n$ is defined via $z(P_n)=\mathbb{H}^+ \subset \mathbb{C}_{\infty}$ while $\mathcal{A}_g$ denotes non singular odd half periods on $\mathcal{J}(\mathcal{M})$ (see Section 6.2) and  $\vartheta$ are  first order Riemann theta functions defined on $\mathcal{J}(\mathcal{M})$. In presence of an odd number $n$ of branes, $n \geq 5$, we can reduce to the $n-1$ even case.  By construction $\theta_{\mathcal{S}_1}(Q,P)=\theta_{\mathcal{S}_2}(P,Q)$, $\theta_{\mathcal{S}_2}(Q,P)=\theta_{\mathcal{S}_1}(P,Q)$ for any number of branes.

\end{Theorem}

The paper is structured as follows. In Section 3 we introduce  the P$\sigma$M, its Batalin Vilkovisky action and  the superpropagators as homotopy operators. We define the concept of  coisotropic branes in Poisson geometry and we show how they  naturally arise in the BV quantization of the P$\sigma$M action describing the one brane case. We compute then the dimension of the space of the "zero modes", i.e. the dimension of the De Rham cohomology of the differential complexes $\mathcal{H}^n_i$. Section 4 and 5 deal with the two and three branes cases: we  want  integral kernels which are zero on the boundary of $P_2$ and $P_3$ for particular sets of indices specified by the presence of branes. To determine such kernels is equivalent to find generalizations of  Kontsevich's angle formula \cite{Kont} with the correct boundary conditions imposed by the constraint equations defining the branes.

 Once defined the  polygon $P_n$ as image under $u$ of $\mathbb{H}^+$ then it is possible to express the boundary conditions of the integral kernels as reflections properties relative to the sides of the polygon $P_n$ of  maps $\psi(u,v)_{\mathcal{S}_i}: P_n \times P_n \rightarrow \mathbb{R}/ 2\pi\mathbb{Z}$: the method is a multi brane generalization of the classical "mirror charges" formalism.
 All the  problem is then reduced to find explicit $\psi(u,v)_{\mathcal{S}_i}$  satisfying the correct reflections. The presence of branes implies the existence of  non zero simple observables even in the trivial Poisson structure case; the complete analysis of the algebra of observables for the P$\sigma$M with branes will appear elsewhere.
 In Section 6 we deal with  $n \geq 4$ branes; we write the $\psi(u,v)_{\mathcal{S}_i}$ maps via  first order Riemann theta functions with odd non singular characteristics and hyperelliptic curves $\mathcal{M}$ \cite{FK}: we refer to this constructions as the canonical setting. With more than three branes we have zero modes contributions: we show that their presence is stricly correlated with the boundary conditions the integral kernels must satisfy.  Section 7 is about   conclusions and further developments; in Appendices A and B we put known material on Riemann theta functions and  the proof of Proposition 1. In Appendix C we describe some properties of the superpropagators as homotopy operators.
 
 This paper is meant to be an introduction to the P$\sigma$M with $n\geq 2$ branes; with the explicit formulas for the superpropagators we can study the algebraic structure and the deformation of the associative  product of the algebra of observables, relating it to the already known results which show $P_{\infty}$ properties for the one brane case. With non trivial Poisson structure it is possible to extend the results of   \cite{Cattaneo:1999fm, Cattaneo:2005zz,Calvo:2005th} for the deformation quantization of branes in the sense of $A_{\infty}$ bimodules.
 In order to study the uniqueness of the superpropagators one can introduce the Laplacian in a suitable metric completion of the differential complexes $\mathcal{H}^n_i$. This topic and the relations between the superpropagators for the $n \geq 4$ cases and classical kernels on compact Riemann surfaces will be discussed in \cite{next}.
  
  The non perturbative study of the P$\sigma$M in presence of general boundary conditions  involves the construction of a generalization of the Fukaya $A_{\infty}$-category:  work in this direction is in progress motivated by the constructions of Section 6 in terms of Riemann surfaces. The idea is to begin producing a local version of the Fukaya category applying the HPT tools to the differential complexes $(\mathcal{H}^n_{\mathcal{S}_i},d)$ to generate an $A_{\infty}$ structure on the cohomology, then define an associated $A_{\infty}$ category. The adjective local refers to the fact that we are in the perturbative contest, expanding around zero modes, in presence of $n\geq 4$ branes.

\section{Acknowledgements}
The idea to use the mirror charges formalism to solve boundary conditions is already  contained in \cite{Cattaneo:2005zz} where Cattaneo and Felder use it to produce the generalized angle maps in the $n=2$ branes case. I would really like  to  thank G. Felder for introducing me this method  and for the many inspiring discussions we had throughout the writing of this paper. I thank also A. Cattaneo, H. Kajiura, C. Rossi for helpful comments  on the initial drafts of this work. This work has been partially supported by the Swiss National Science Fundation (grant 200020-105450).

\section{The P$\sigma$M: classical action functional and BV quantization }
The Poisson Sigma Model \cite{Schaller,Ikeda} is a two-dimensional topological Sigma  theory defined on a two dimensional orientable surface $\Sigma$ and with  a Poisson manifold $(M,\pi)$ as target. It is defined by a classical functional $S$ on the space of bundle maps $(X,\eta): T\Sigma \rightarrow T^*M$ with base map $X: \Sigma \rightarrow M$ and $\eta \in \Omega(\Sigma, X^*T^*M)$, where $S$ is given explicitly by:

\begin{equation}
\label{poisson}
S[X,\eta]=\int_{\Sigma} \langle \eta, d X \rangle  +\frac{1}{2}\langle \pi \circ X,\eta \wedge \eta \rangle,
\end{equation}
and $\langle,\rangle$ is the canonical pairing between vectors and covectors of $M$. 
The Euler-Lagrange equations express the condition that the pair $(X,\eta)$ is a Lie algebroid morphism between $T\Sigma$  and $T^*M$. 
Under the infinitesimal gauge transformations $\delta_{\beta}X^i=\pi^{ij}(X)\beta_j$, $ \delta_{\beta}\eta_i=-d\beta_i-\partial_i\pi^{jk}(X)\eta_j\beta_k$, where $\beta=\beta_idX^i$ is a section of $X^*(T^*M)$, the action (\ref{poisson}) transforms by a boundary term $\delta_{\beta}S=-\int_{\Sigma}d(dX^i\beta_i)$. The commutator of two infinitesimal gauge transformations is a gauge transformation only on shell, that is when Euler-Lagrange equations are fulfilled for the action $S$. Thus the gauge transformations form a Lie algebra only when acting on the set of critical points, or classical solutions of $S$: in order to quantize the model we need to implement the Batalin-Vilkovisky (BV) procedure; before this one defines the  BRST operator $\delta_0$ \cite{Cattaneo:1999fm}  (the gauge parameters $\beta_i$ are promoted to ghost fields ) and the  boundary conditions for the model \cite{CalF}. 

\subsection{Perturbative analysis: space filling brane case}

Varying the action (\ref{poisson}) we observe the appearance of a boundary term of the form

\begin{eqnarray}
\label{bound2}
\int_{\partial\Sigma}\langle \eta, \delta X\rangle.
\end{eqnarray}
 If we assume that $\Sigma:=\{z \in \mathbb{C}  \mid~ \mid z \mid \leq 1 \}$ with $\eta$ vanishing when contracted with vectors tangent to the boundary, then (\ref{bound2}) cancels; imposing that  the infinitesimal parameters $\beta_i$ vanish on  $\partial\Sigma$ we cancel $\delta_{\beta}S$. We refer to these BC as the "space filling brane" case \cite{Cattaneo:1999fm}.
We introduce some notation concerning the BV formalism; we refer to \cite{Cattaneo:1999fm} for the full analysis and to \cite{Super} for a geometrical approach.  We introduce the "antifields" $\phi^+ := (X^+,\eta^+,\beta^+)$ with complementary ghost number and degree as differential forms on $\Sigma$ w.r.t. the "fields"  $\phi := (X,\eta,\beta)$. Then we look for the so called Batalin-Vilkovisky action $S_{BV}[\phi,\phi^+]$ of ghost number zero such that $S_{BV}[\phi,0]$ reduces to the classical action $S[\phi]$ and $(S_{BV},S_{BV})-2\hbar i \Delta S_{BV}=0$, where $(,)$ is the BV antibracket and $\Delta$ is the BV Laplacian.  

Then we fix the  gauge $d*\eta_i=0$. The Hodge operator $*$ is explicitly given, in terms on the coordinates on $\Sigma$, by $*dx^1=dx^2,$ and $*dx^2=-dx^1$ with $z \in \Sigma, z=x_1+ix_2$ and the gauge fixing fermion is $\Psi =-\int_{\Sigma} d\gamma^i *\eta_i$. In order to express the BV Laplacian one can introduce the Hodge dual antifields $\phi^*_{\alpha}:= *\phi_{\alpha}$ with the rule that they must have the same boundary conditions as the fields.  Selecting the Lagrangian submanifold $\mathcal{L}$ defined by the equations $\phi^+_{\alpha}=\frac{\partial}{\partial\phi_{\alpha}}\Psi$ and adding the antighost term $-\int \lambda^i\gamma^+_i$, then the gauge fixed action is:

\begin{eqnarray*}
S_{gf}&=&\int_{\Sigma} \eta_i \wedge d X^i +\frac{1}{2}\pi^{ij}(X)\eta_i \wedge \eta_j-*d\gamma^i \wedge (d\beta_i+\partial_i\pi^{kl}(X)\eta_k\beta_l)+\nonumber \\
&&-\frac{1}{4}*d\gamma^i \wedge *d\gamma^j  \partial_{i}\partial_j \pi^{kl}(X)\beta_k\beta_l -\lambda^id*\eta_i,
\end{eqnarray*}
where, in particular, $\eta_i^+=*d\gamma^i$ and the Lagrange multipliers $\lambda^i$  satisfy Dirichlet boundary conditions on $\partial\Sigma$. In the perturbative expansion of the space filling brane case we select $\Sigma=D$, $D$ unit disk in $\mathbb{C}$ and we expand around the classical solution  $X(x)=y$, $\eta=0$, i.e. $X(x)=y+\xi(x)$, $\xi(x)$ fluctuation field. The  kinetic part of the gauge fixed action:

\begin{eqnarray*}
S_{gf}^0&=&\int_{\Sigma} \eta_i \wedge d \xi^i -\lambda^id*\eta_i-*d\gamma^i \wedge d\beta_i.
\end{eqnarray*}
We map conformally the disk  to the (compactified) complex upper half plane $\mathbb{H}^{+}$ and we use the standard complex coordinates $(z,w)$ on it: introducing  $\tilde{\xi} := \xi-*d\gamma, \tilde{\eta}:= \beta+\eta$ we obtain  what in \cite{Cattaneo:1999fm} are called the "superpropagators" for the space filling brane case:

 \begin{eqnarray}
\label{sp} 
 ~~~~G(w,z)_{\mathcal{A}} :=
\langle \tilde{\xi}^k(w)\tilde{\eta}_j(z) \rangle= \frac{i\hbar}{2\pi}\delta_{\mathcal{A}}d\phi(z,w),
\end{eqnarray}
where $d=d_z+d_{w}$,  $\phi(z,w):=\frac{1}{2i}\ln \frac{(z-w)(z-\bar{w})}{(\bar{z}-\bar{w})(\bar{z}-w)}=\arg(z-w)+\arg(z-\bar{w})$ is  Kontsevich's angle function \cite{Kont} with $d_z=dz\frac{\partial}{\partial z}+d\bar{z}\frac{\partial}{\partial \bar{z}}$ and $k,j \in \mathcal{A}=\{1,2\dots,m\}$, $m=$dim$M$. The angle function $\phi: \mathbb{H}^+ \times \mathbb{H}^+ \rightarrow \mathbb{R}/2\pi\mathbb{Z}$ associates to each pair of distinct points $(z,w)$ in the upper half plane the angle between the geodesics w.r.t. the Poincar\'e metric connecting $z$ to $+i\infty$ and to $w$, measured in the counterclockwise direction.  On the boundary $\partial \mathbb{H}^+=\mathbb{R}\cup \{\infty \}$ we have   $\phi( z \in \mathbb{R}\cup \{\infty \},w)=0$.

 \subsection{Branes as generalized boundary conditions}A submanifold of the Poisson manifold $M$ chosen as boundary condition  is called brane; symmetries of the model give strict characterizations of branes in terms of  Poisson geometry. A bunch of definitions;  the sharp map $\pi^{\sharp}$ for the Poisson manifold $(\mathcal{M},\pi)$ is defined as $\pi^{\sharp}: T^*M \rightarrow TM$, with $\langle \pi^{\sharp}(y)\sigma,\tau \rangle=\pi(y)(\sigma,\tau), \forall y \in M$, $\forall \sigma,\tau \in T^*_yM$ and $\langle,\rangle$ denotes the canonical pairing.  A submanifold $\mathcal{C}$ of the Poisson manifold $(M,\pi)$ is called pre-Poisson \cite{CZ} if $\pi^{\sharp}(N^*\mathcal{C}) + T\mathcal{C}$ has constant rank along $\mathcal{C}$: in \cite{CalF} this was called a "submanifold with strong regular conditions". In the symplectic context this condition is equivalent to $\mathcal{C}$ being presymplectic. In \cite{Calvo:2005th} it is shown that if $\mathcal{C}$ is pre-Poisson then $A\mathcal{C}=\pi^{\sharp-1}T\mathcal{C}\cap N^*\mathcal{C}$ is a Lie subalgebroid of the full Lie algebroid $T^*M$. A submanifold $\mathcal{C} \subset M$ is called coisotropic if $\pi^{\sharp}(N^*\mathcal{C}) \subset T\mathcal{C}$. It follows from the Jacobi identity for $\pi$ that the characteristic distribution $\pi(N^*\mathcal{C})$ on the coisotropic submanifold $\mathcal{C}$ is involutive; the corresponding foliation is called the characteristic foliation. In the symplectic context, $\pi^{\sharp}$ yields an isomorphism between $N^*\mathcal{C}$ and $T^{\perp}\mathcal{C}$ and we recover the usual definition of coisotropic submanifolds in the symplectic case: $T^{\perp}\mathcal{C} \subset T\mathcal{C}$, where $T^{\perp}\mathcal{C}$ is the subbundle of $T_{\mathcal{C}}M$ of vectors that are symplectic orthogonal to all vectors of $T\mathcal{C}$. If $\mathcal{C}$ is coisotropic, then $A\mathcal{C}=N^*\mathcal{C}$.  Pre-Poisson submanifolds are the most general boundary conditions for P$\sigma$M compatible with symmetries \cite{C2007}; it can be  shown (cfr. \cite{CZ}) that if $\mathcal{C}$ is Pre-Poisson in $M$ then it is always possible to find a cosymplectic submanifold of $M$ which contains $\mathcal{C}$ as a coisotropic submanifold: then it is enough to consider coisotropic submanifolds as boundary conditions for the P$\sigma$M. This means that, in the one brane case, given a coisotropic submanifold $\mathcal{C}$  of $M$ (i.e the brane),  we impose the boundary conditions $X|_{\partial\Sigma}: \partial\Sigma \rightarrow \mathcal{C}$, $i_{\partial\Sigma}^*\eta \in \Gamma( X^*N^*\mathcal{C})$, where $i_{\partial\Sigma}$ is the inclusion map $\partial\Sigma \hookrightarrow \Sigma$ and the ghosts on the boundary satisfy $i_{\partial\Sigma}^*\beta \in \Gamma(X^*N^*\mathcal{C})$.

 \subsection{Perturbative analysis: one brane case... }  The perturbative analysis in presence of a brane is performed considering the case where $M$ is an open subset of $\mathbb{R}^n$ with coordinates $x^1,\dots,x^n$ and the brane $\mathcal{C}$ is given by the constraint equations $x^{\mu}=0, \mu \in I ,$ with $\mu=m+1,\dots,n$. The tangent space to a point of $\mathcal{C}$ is spanned by $\frac{\partial}{\partial x^{i}}$, $i=1,\dots,m$ and the conormal bundle  by $dx^{\mu}$, $\mu=m+1,\dots,n$. We follow the convention that latin indices run along the brane, that is over $\left\{1,\dots,m \right\}$,  while  greek indices are associated to coordinates normal to  the brane and run over $\left\{m+1,\dots,n \right\}$. The whole boundary of the disk, (or conformally, the compactified real line in the upper half plane), is mapped to the brane $\mathcal{C}$:  the splitting of the indices into $\mathcal{S}_1:=I^c$ and $\mathcal{S}_2:=I$ induces   $G_{\mathcal{S}_1}(\omega,z)$ and $ G_{\mathcal{S}_2}(\omega,z)$ and the  boundary conditions for $\xi$ and $\eta$ become  $\xi^{\mu}|_{\mathcal{C}}=0$ and $i^*_{\mathcal{C}}\eta_{i}=0$. Consequently the "superpropagators" with the correct boundary conditions are \cite{CF}:

\begin{eqnarray}
G_{\mathcal{S}_1}(w,z)= \frac{i\hbar}{2\pi}\delta_{\mathcal{S}_1}d\phi(z,w), ~~G_{\mathcal{S}_2}(w,z)= \frac{i\hbar}{2\pi}\delta_{\mathcal{S}_1}d\phi(w,z), \label{ppp2}
\end{eqnarray}
  where $\phi$ is  Kontsevich's angle function and $\delta_{\mathcal{S}_i}$ is the Kronecker delta restricted to pairs of indices in $_{\mathcal{S}_i}$. This result allows to study the quantization of the coisotropic brane $\mathcal{C}$ via path integral, with the introduction in the diagrammatics of straight and wavy lines  induced by the presence of the two   index sets $\mathcal{S}_1$ and $\mathcal{S}_2$.

\subsection{...and two or more branes cases.}
In presence of two or more branes a more general analysis of the boundary conditions is needed. The branes are defined as usual by the constraints 

\begin{equation}
\label{branes}
\mathcal{C}_j=\left\{ x^{\mu_j}=0 \mid \mu_j \in I_j  \right\}
\end{equation}
and given $n \geq 2$ branes (\ref{branes}) we define the index sets:

\begin{eqnarray}
&&\mathcal{S}_1 :=\left\{ \begin{array}{c} I_1^c\cap I_2 \cap I^c_3\cap \dots \cap I_n  ~~\mbox{ $n$ even} \\ I_1^c\cap I_2 \cap I^c_3\cap \dots \cap I^c_n  ~~\mbox{$n$ odd}, \end{array} \right.  \nonumber \\
&&\mathcal{S}_2:=\left\{ \begin{array}{c} I_1\cap I_2^c \cap I_3\cap \dots \cap I_n^c  ~~\mbox{$n$ even} \\  I_1\cap I_2^c \cap I_3\cap \dots \cap I_n ~~\mbox{$n$ odd}  \end{array}\right. \label{index!}
\end{eqnarray}
Let $\Sigma:=P_n$ be the source manifold for the P$\sigma$M in presence of $n\geq 2$ branes.  $P_n$ is defined as a $n$ sided convex polygon with boundary $\partial\Sigma = \partial P_n$ given by the decomposition $\partial \Sigma = \cup_{i=1}^{n} \partial P_n^i$. The corners of the polygon are the elements of  the set $\{ \partial P_n^i \cap \partial P_n^{i+1} \}_{i=1,\dots,n ~\mbox{\tiny{mod n}}}$. The polygon $P_n$ is fixed by the condition $P_n :=u(\mathbb{H}^+)$, where   $u:\mathbb{H}^+ \rightarrow P_n$ is a homeomorphism depending by the number $n$ of branes considered in the perturbative analysis and $\mathbb{H}^+$ denotes the compactified complex upper half plane. In particular this allows to give an ordering to the sides $\partial P_n^i$ in a consistent way. All the possible index sets $\mathcal{A}_k$, $k=1,\dots,2^n$ for the integral kernels 

\begin{equation}\label{intker}
\theta(Q,P)_{\mathcal{A}_k} := -\frac{i}{\hbar}\langle \tilde{\xi}^{a_1}(Q)\tilde{\eta}_{a_2}(P) \rangle,
\end{equation}
with $(P,Q) \in P_n \times P_n, a_1,a_2 \in \mathcal{A}_k$ are given by the intersections of the sets $I_j,I^c_j$ defining the $n$ branes (\ref{branes}).
In the sequel we will drop out the $\delta^{a_2}_{a_1}$  dependence in (\ref{intker}); we will reduce to the $a_1=a_2$ case to simplify notation.
 Fixing the number of branes $n$, one repeats the same calculations of the preceding subsections, giving new boundary conditions to eliminate the additional terms
 $\int_{\partial\Sigma} \langle\eta,\delta X\rangle, \int_{\partial\Sigma}\beta dX $ and to develop a coherent BV formalism. The boundary conditions for the $X$ field are given by $X|_{\partial P_n^i }:\partial P_n^i \rightarrow \mathcal{C}_i$ while we choose $\eta|_{\partial P_n^i} \in \Gamma (X^*N^*\mathcal{C}_i)$, $\beta|_{\partial P_n^i} \in \Gamma (X^*N^*\mathcal{C}_i)$ $i=1,\dots,n$ generalizing the one brane case. More explicitly, $\eta|_{\partial P_n^i}:= i^{*}_{\partial P_n^i}\eta$, $\beta|_{\partial P_n^i}:= i^{*}_{\partial P_n^i}\beta$ where $i_{\partial P_n^i}:\partial P_n^i \hookrightarrow P_n$ denotes the inclusion. We apply the same notation to all the component fields appearing in the sequel.
 $\beta$ satisfies the same boundary conditions of $\eta$ as it belongs to the same superfield $\tilde{\eta}$ (we refer to \cite{Cattaneo:1999fm} for a brief introduction to the superfield formalism): in particular this  allows to cancels the boundary terms coming from $-\int_{\Sigma}d(dX^i\beta_i)$).
At the same time $*d\gamma=\eta^+$ belongs to the same superfield $\tilde{X}$ of $X$.

 In the sequel we will refer to the sides $\partial P_n^i$ as branes, whenever confusion does not arise with the corresponding (under $X$) coisotropic submanifolds $\mathcal{C}_i \subset (M,\pi)$.
 For simplicity let us discuss the boundary conditions in the $n=2$ branes case explicitly, i.e. $\partial\Sigma=\partial P_2^1 \cup \partial P_2^2$,  $\mathcal{A}_1=I_1 \cap I_2$, $\mathcal{A}_2=I_1^c \cap I_2^c$, $\mathcal{S}_1=I_1^c \cap I_2$, $\mathcal{S}_2=I_1 \cap I_2^c$. 
As $\delta X=\pi^{\sharp}\beta$ then $\delta X\mid_{\partial P_2^i}\in T\mathcal{C}_i$, 
$X^{\mathcal{A}_1}\mid_{\partial P_2^1}=X^{\mathcal{S}_2}\mid_{\partial P_2^1}=X^{\mathcal{A}_1}\mid_{\partial P_2^2}=X^{\mathcal{S}_1}\mid_{\partial P_2^2}=0$ and  $\eta^{\mathcal{A}_2}\mid_{\partial P_2^1}=\eta^{\mathcal{S}_1}\mid_{\partial P_2^1}=\eta^{\mathcal{A}_2}\mid_{\partial P_2^2}=\eta^{\mathcal{S}_2}\mid_{\partial P_2^2}=0$. In particular we get

 \begin{eqnarray*}
   X^{\mathcal{S}_1}\mid_{\partial P_2^2}=0, \eta^{\mathcal{S}_1}\mid_{\partial P_2^1}=0,   \\     
 X^{\mathcal{S}_2}\mid_{\partial P_2^1}=0, \eta^{\mathcal{S}_2}\mid_{\partial P_2^2}=0,
 \end{eqnarray*}
 i.e. specifing the index sets $\mathcal{S}_i$ we have Dirichlet boundary conditions for the components fields of $(\tilde{X},\tilde{\eta})$ only on half of the boundary $\partial P_2$.  A general fact: given $n$ number of branes, then the  "non trivial" integral kernels are those associated to the alternating Dirichlet boundary conditions respect to the points $(P,Q)$ of the polygon $P_n$, with index sets given by $\mathcal{S}_1,\mathcal{S}_2$.
In presence of all the other index sets, it is possible to "reduce" to a lower number of branes case to compute the superpropagators.

 Developing the BV formalism we generalize the gauge fixing condition $d*\eta=0$ to eliminate the boundary terms $\int_{\partial P_2} \lambda^i*\eta_i=\int_{\partial P_2^1} \lambda^{i_1}*\eta_{i_1}\mid_{i_1\in\mathcal{S}_2}+
\int_{\partial P_2^1} \lambda^{i_2}*\eta_{i_2}\mid_{i_2\in\mathcal{A}_1}+\int_{\partial P_2^2} \lambda^{i_3}*\eta_{i_3}\mid_{i_3\in\mathcal{S}_1} +\int_{\partial P_2^2} \lambda^{i_4}*\eta_{i_4}\mid_{i_4\in\mathcal{S}_2}$ imposing the extended gauge fixing

\begin{eqnarray}
d*\eta&=&0, \nonumber \\
*\eta^{\mathcal{S}_1}\mid_{\partial P_2^2}&=&*\eta^{\mathcal{S}_2}\mid_{\partial P_2^1}=0 \label{egf},
\end{eqnarray}
and $*\eta^{\mathcal{A}_1}\mid_{\partial P_2^2}=*\eta^{\mathcal{A}_1}\mid_{\partial P_2^2}=0$ with boundary conditions for the Lagrange multipliers given by $\lambda^{\mathcal{S}_1}\mid_{\partial P_2^1}=\lambda^{\mathcal{S}_2}\mid_{\partial P_2^2}=\lambda^{\mathcal{A}_2}\mid_{\partial P_2^1}=\lambda^{\mathcal{A}_2}\mid_{\partial P_2^2}=0$. The extended gauge fixing also imposes $d\gamma^{\mathcal{S}_1}\mid_{\partial P_2^1}=d\gamma^{\mathcal{S}_2}\mid_{\partial P_2^2}= d\gamma^{\mathcal{A}_2}\mid_{\partial P_2^1}=d\gamma^{\mathcal{A}_2}\mid_{\partial P_2^2}=0$. With the choice (\ref{egf}) we fix the boundary conditions for the components fields of the superfields $(\tilde{X},\tilde{\eta})$ on the whole $\partial P_2$ for the index sets $\mathcal{S}_i$, getting  alternating Dirichlet-Neumann boundary conditions.

The general $n$ branes case follows the same lines of the $n=2$ case; the Dirichlet boundary conditions and the extended gauge fixing are collected in the following table.

\begin{center}
\begin{tabular}{|r|}\hline
Dirichlet Boundary conditions in the $n$ branes case \\ \hline
 $X^{\mathcal{S}_1}\mid_{\partial P_n^{even}}=0, \eta^{\mathcal{S}_1}\mid_{\partial P_n^{odd}}=0$~~~~~~~~~~~~~~~~ \\ 
 $X^{\mathcal{S}_2}\mid_{\partial P_n^{odd}}=0, \eta^{\mathcal{S}_2}\mid_{\partial P_n^{even}}=0$~~~~~~~~~~~~~~~~ \\  \hline
Extended gauge fixing in the $n$ branes case \\ \hline
 $d*\eta=0$~~~~~~~~~~~~~~~~~~~~~~~~~~~~ \\
 $*\eta^{\mathcal{S}_1}\mid_{\partial P_n^{even}}=0, *\eta^{\mathcal{S}_2}\mid_{\partial P_n^{odd}}=0$~~~~~~~~~~~~~~~ \\ \hline
\end{tabular}
\end{center}
This implies that the integral kernels (\ref{intker}) satisfy

\begin{eqnarray}\label{bc}
\theta(Q,P \in \partial P_{n}^{odd} )_{\mathcal{S}_1}&=&0, \nonumber \\ 
\theta(Q \in \partial P_{n}^{even},P  )_{\mathcal{S}_1}&=&0, \nonumber \\
\theta(Q,P \in \partial P_{n}^{even} )_{\mathcal{S}_2}&=&0, \nonumber \\ 
\theta(Q \in \partial P_{n}^{odd},P  )_{\mathcal{S}_2}&=&0, 
\end{eqnarray} 
with $\mathcal{S}_i$ given by (\ref{index!}).

\subsection{Superpropagators and homotopy operators }

The presence of branes in the perturbative expansion induces the "alternating" boundary conditions on the component fields of the superfields $(\tilde{X},\tilde{\eta})$ as shown in the preceding section. This allows to give the following

\begin{Def}{\bf{Relevant superpropagators for the P$\sigma$M in presence of branes}}
The   gauge fixed homotopy operators $G_{\mathcal{S}_i}$
which realize a Hodge-Kodaira splitting

\begin{equation}\label{splitting}
dG_{\mathcal{S}_i}+G_{\mathcal{S}_i}d=I-P_{\mathcal{S}_i}
\end{equation}
 of the differential complexes  $(\mathcal{H}_{\mathcal{S}_i}^n,d):= \left\{ \begin{array}{c} \Omega(P_n,\partial P^{even}_n) ~i=1 \\ \Omega(P_n,\partial P^{odd}_n) ~~i=2 \end{array} \right.$ are called relevant superpropagators for the P$\sigma$M in presence of $n$ branes (\ref{branes}). With   $P_{\mathcal{S}_i}$ we  denote a projection onto cohomology $H(\mathcal{H}_{\mathcal{S}_i}^n)$. The $G_{\mathcal{S}_i}$ operators act via

\begin{equation}\label{operator}
(G_{\mathcal{S}_i}\phi)(Q):=\int_{P_n} \theta(Q,P)_{\mathcal{S}_i} \wedge \phi(P) ~~~~Q \in P_n,
\end{equation}
$\forall \phi \in (\mathcal{H}_{\mathcal{S}_i}^n,d)$. The integral kernels (\ref{intker})  are given explicitly by the sum 

\begin{equation}\label{intkernn}
\theta(Q,P)_{\mathcal{S}_i}:=\frac{1}{2\pi}d \arg \psi(P,Q)_{\mathcal{S}_i}-\mathcal{Z}_{\mathcal{S}_i}(Q,P),
\end{equation}
 where
\begin{equation*}
 \mu_{\mathcal{S}_i}(P,Q):= \arg \psi(P,Q)_{\mathcal{S}_i} 
\end{equation*}
is called generalized angle function and  $d\mathcal{Z}_{\mathcal{S}_i}(Q,P)=\mathcal{P}(Q,P)_{\mathcal{S}_i}$ is the integral kernel of the projection $P_{\mathcal{S}_i}$ onto cohomology, $d=d_P+d_Q$. The generalized angle functions satisfy $d\mu_{\mathcal{S}_i}(P,Q) \sim d\arg (z_P-w_Q)$ for $P \rightarrow Q$, where $(z_P,w_Q)$ are coordinates of the points $(P,Q)$.
\end{Def}
 
The explicit form of the integral kernels is simply given by a contribution which generalizes the Kontsevich angle function due to the boundary conditions imposed by the presence of the branes and  a non trivial term due to  de Rham cohomology $H(\mathcal{H}_{\mathcal{S}_i}^n)$. With an explicit choice of metric (i.e. a Hodge star operator) it is possible to introduce a Laplacian operator on a metric completion  the differential complexes $(\mathcal{H}_{\mathcal{S}_i}^n,d)$ (and a gauge fixing for the theory, as we have already seen). 
Defining   harmonic forms on the metric completion one can study the uniqueness of the relevant superpropagators; we will discuss this elsewhere \cite{next}.

\begin{Lem}
Let $(\mathcal{H}_{\mathcal{S}_i}^n,d)$ $i=1,2$ be as above; then the De Rham cohomologies 
$H^{\bullet}(\mathcal{H}^n_{\mathcal{S}_i})$ are given by:
\begin{eqnarray*}
&&H^p(\mathcal{H}_{\mathcal{S}_i}^n) \simeq \left\{ \begin{array}{c} \mathbb{C}^{\frac{n-2}{2}} \null ~~~p=1 \\ 0 \null  ~~~\mbox{otherwise} \end{array} \right. \mbox{for n even},\\
&&H^p(\mathcal{H}_{\mathcal{S}_i}^n) \simeq \left\{ \begin{array}{c} \mathbb{C}^{\frac{n-3}{2}} \null ~~~p=1 \\ 0 \null ~~~\mbox{otherwise} \end{array} \right.  \mbox{for n odd}.
\end{eqnarray*}
\end{Lem}

\begin{proof}
We do the analysis for $i=1$; the other case is equivalent. We write the short exact sequence $0 \hookrightarrow \ker i^* \hookrightarrow \Omega^{\bullet}(P_n) \rightarrow \im i^* \rightarrow 0$, where $\ker i^*:=\Omega^{\bullet}(P_n,\partial P_n^{even})$, $\im i^* := \Omega^{\bullet}(\partial P_n^{even})$ and $i: \sqcup_{i=1}^{\frac{n}{2}} \partial P_n^{2i} \hookrightarrow P_n$  for $n$ even or $i: \sqcup_{i=1}^{\frac{n-1}{2}} \partial P_n^{2i} \hookrightarrow P_n$  for $n$ odd are the inclusions of the even branes in $P_n$. This sequence induces a long exact sequence in cohomology; a standard counting  gives the thesis.
\end{proof}

\section{Two branes case}

In presence of two  branes $\mathcal{C}_1=\left\{ x^{\mu_1}=0 \mid \mu_1 \in I_1  \right\}$ and $\mathcal{C}_2=\left\{ x^{\mu_2}=0 \mid \mu_2 \in I_2  \right\}$ we can select the indices for the superpropagators in the sets $\mathcal{A}_1:= I_1 \cap I_2$, $\mathcal{A}_2:=I_1^c \cap I_2^c $ and $\mathcal{S}_1:=I_1^c \cap I_2 , \mathcal{S}_2:= I_1 \cap I_2^c  $ following (\ref{index!}). The $n=2$ sided polygon $P_2$ is defined as $P_2 :=u(\mathbb{H}^+)$, with the  Schwarz-Christoffel  mapping $u$ \cite{Gonzal} given by  

\begin{eqnarray*}
\label{M23}
z \rightarrow  u(z) := \sqrt{z}.
\end{eqnarray*}
Points $(P,Q) \in P_2 \times P_2$ are represented respectively by a pair of complex numbers $(u,v)$ in the first quadrant, with $u=u(z), v:=u(w)$ $~\forall (z,w) \in \mathbb{H}^+ \times \mathbb{H}^+$. $\partial P_2^1$ is given by the positive imaginary axis, while $\partial P_2^2$ is the positive real  axis.

The boundary conditions imposed by the index sets $\mathcal{S}_i$ are  $\theta(v,u\in\partial P_2^1 )_{\mathcal{S}_{1}}=\theta(v\in \partial P_2^2,u)_{\mathcal{S}_{1}}=0$,  $\theta(v,u\in\partial P_2^2)_{\mathcal{S}_{2}}=\theta(v\in\partial P_2^1,u)_{\mathcal{S}_{2}}=0$. 
Introducing the  maps

\begin{eqnarray}
\psi(u,v)_{\mathcal{S}_{1}}= \arg \frac{(u-v)(\bar{u}-v)}{(\bar{u}+v)(u+v)},~~~\psi(u,v)_{\mathcal{S}_{2}}=\arg \frac{(u-v)(\bar{u}+v)}{(\bar{u}-v)(u+v)}, \label{M2}
\end{eqnarray}
which satisfy the same boundary conditions of $\theta(v,u)_{\mathcal{S}_i}$ and considering that $H(\mathcal{H}^2_i)=\{0\}$, we get

\begin{Theorem}
The integral kernels for the  superpropagators $G_{\mathcal{S}_i}$ in presence of two branes are given by

\begin{equation*}
\label{zxc}
\theta(v,u)_{\mathcal{S}_i}= \frac{1}{2\pi}d\psi(u,v)_{\mathcal{S}_{i}},
\end{equation*}
with mirror maps  (\ref{M2}). The integral kernels  satisfy the additional boundary conditions $ \theta(v,u)_{\mathcal{S}_1}=\theta(v,\bar{u})_{\mathcal{S}_1}=\theta(-\bar{v},u)_{\mathcal{S}_1},$ $\theta(v,u)_{\mathcal{S}_2}=\theta(v,-\bar{u})_{\mathcal{S}_2}=\theta(\bar{v},u)_{\mathcal{S}_2}$, i.e. every boundary component of $P_2$ is labelled by a boundary condition for both the variables $(u,v)$. By construction $\theta(v,u)_{\mathcal{S}_1}=\theta(u,v)_{\mathcal{S}_2}$, $\theta(v,u)_{\mathcal{S}_2}=\theta(u,v)_{\mathcal{S}_1}$. 
 
\end{Theorem}

One can verify that $\theta(v,u)_{\mathcal{A}_1}= \frac{1}{2\pi}d\arg \frac{(u-v)(u+v)}{(u+\bar{v})(u-\bar{v})}=\frac{1}{2\pi}d\arg\frac{(z-w)}{(z-\bar{w})}=
\frac{1}{2\pi}d\phi(z,w)$ and  $\theta(v,u)_{\mathcal{A}_2}= \frac{1}{2\pi}d\arg \frac{(u-v)(u+v)}{(\bar{u}-v)(\bar{u}+v)}=\frac{1}{2\pi}d\arg\frac{(z-w)}{(\bar{z}-w)}=\frac{1}{2\pi}d\phi(w,z)$
with $\phi$  Kontsevich's angle function.

 The integral kernels here presented  correspond to the generalized angle functions in \cite{CF,CT};  they are used to construct an explicit quantum deformation of   bimodule structures  and to define a Kontsevich's product (associative!) in presence of two branes. Additional assumptions are necessary to take care of the faces produced in the compactification of some configuration spaces in order to guarantee associativity (see \cite{CF} ). The problem of finding integral kernels with correct boundary conditions is replaced by the easier task to write the $\psi(u,v)_{\mathcal{S}_{i}}$ maps satisfying some reflection properties  respect to the sides  $\partial P_n^j$: the $\psi(u,v)_{\mathcal{S}_{i}}$ are odd respect to these reflections, allowing  to determine the correct  kernels. The method is nothing but a generalization of the classical "mirror charges" formalism, due to the presence of multiple axis of symmetry, or branes.

\section{Three branes case}

In presence of three branes $\mathcal{C}_1, \mathcal{C}_2, \mathcal{C}_3$ we can select 8 different index sets. We define as usual $\mathcal{S}_1 :=I_1^c \cap I_2 \cap I_3^c$, $\mathcal{S}_2:=I_1 \cap I_2^c \cap I_3$;  the three sided polygon $P_3$ is defined via $P_3:=u(\mathbb{H}^+)$ with

\begin{eqnarray*}
\label{Mirror3}
z \rightarrow u(z) := \frac{1}{2\pi} \int_1^z \frac{ds}{\sqrt{s(s-1)}}, \forall z \in \mathbb{H}^+.
\end{eqnarray*}
where the integral is performed along a smooth path in $\mathbb{H}^+$. $P_3$ is a strip in the first quadrant with two sides parallel to the real axis. The above side is $\partial P_3^1$; the other sides are labelled counterclockwise. The boundary conditions for the integral kernels of the relevant superpropagators are $\theta(v,u\in\partial P_3^3)_{\mathcal{S}_1}=  \theta(v\in\partial P_3^2,u)_{\mathcal{S}_1}= \theta(v,u\in\partial P_3^1)_{\mathcal{S}_1}=0$, $\theta(v\in\partial P_3^1,u)_{\mathcal{S}_2}=  \theta(v,u\in\partial P_3^2)_{\mathcal{S}_2} = \theta(v\in\partial P_3^3,u)_{\mathcal{S}_2}$. 

Using the function $F(u-v) := \arg \sin i\pi (u-v)$ we write the mirror maps:

\begin{eqnarray}
\psi(u,v)_{\mathcal{S}_1}=\arg \frac{\sin i\pi (u-v)\sin i\pi (\bar{u}+v)}{\sin i\pi (\bar{u}-v)\sin i\pi (u+v)}, \nonumber \\ 
\psi(u,v)_{\mathcal{S}_{2}}=\arg \frac{\sin i\pi (u-v)\sin i\pi (\bar{u}-v)}{\sin i\pi (\bar{u}+v)\sin i\pi (u+v)}. \label{po2}
\end{eqnarray}
So we get (remembering that $H(\mathcal{H}_i^3))=\{0\}$)

\begin{Theorem}
The integral kernels for the relevant superpropagators $G_{\mathcal{S}_i}$  in presence of three branes  are given by 

\begin{equation*}
\label{three}
\theta(v,u)_{\mathcal{S}_i}= \frac{1}{2\pi}d\psi(u,v)_{\mathcal{S}_{i}},
\end{equation*}
with mirror maps given by (\ref{po2}). The integral kernels satisfy the additional boundary conditions $ \theta(v,u)_{\mathcal{S}_1}=\theta(v,-\bar{u})_{\mathcal{S}_1}=\theta(\bar{v},u)_{\mathcal{S}_1}=\theta(i+\bar{v},u)_{\mathcal{S}_1},$ $\theta(v,u)_{\mathcal{S}_2}=\theta(v,\bar{u})_{\mathcal{S}_2}=\theta(v,i+\bar{u})_{\mathcal{S}_2}=\theta(-\bar{v},u)_{\mathcal{S}_2}$, i.e. every boundary component of $P_3$ is labelled by a boundary condition for both the variables $(u,v)$. By construction $\theta(v,u)_{\mathcal{S}_1}=\theta(u,v)_{\mathcal{S}_2}$, $\theta(v,u)_{\mathcal{S}_2}=\theta(u,v)_{\mathcal{S}_1}$.
\end{Theorem}

Choosing the six "non relevant" index sets, we can reduce to a lower number of branes analysis. As we impose boundary conditions respect the same variable on adiacent sides of $P_3$ (not separated by $\{ \infty \}$, except in the $\mathcal{A}_1=I_1 \cap I_2\cap I_3$ and $\mathcal{A}_2=I_1^c \cap I_2^c \cap I_3^c$ cases), then it is simple to see that we can introduce a new "square root" homeomorphim  $\tilde{u}$ and reduce to a two branes case. Choosing
$\mathcal{A}_1, \mathcal{A}_2$ and repeating the same lines we can reduce to the one brane case: the interesting boundary conditions are those associated to the index sets (\ref{index!}).

\section{More than three branes: the hyperelliptic cases}

In presence of $n \geq 4$ branes we deal with the so called hyperelliptic cases. The name comes from the fact that we are going to use hyperelliptic Riemann surfaces $\mathcal{M}$ of genus $g$ in order to find the superpropagators. In principle one could write a  Schwarz-Christoffel  mapping from the unit disk with boundary partitioned into $n \geq 4$ sectors, or conformally, from $\mathbb{H}^+$ into a suitable polygon with $n$ sides; this would permit to avoid Riemann surfaces and theta functions formalism. The main problem concerning this formulation is that the reflections to impose on the $\psi(u,v)_{\mathcal{S}_i}$ maps would become particularly complicated and it is not clear a priori which functions one should  use to get the correct  maps fulfilling the $n$ boundary conditions of the integral kernels. The canonical setting in presence of hyperelliptic Riemann surfaces is a natural choice, instead. We refer to Appendix B for a brief introduction on Riemann theta functions. 
Let $\mathcal{M}$ be the hyperelliptic Riemann surface of genus $g$ which realizes the two sheeted branched covering $z:\mathcal{M} \rightarrow \mathbb{C}_{\infty}$ such that $z(P_i):=x_i \in \mathbb{R}, x_i <x_{i+1}$, $\forall i=1,\dots,2g+1$ and $z(P_{2g+2})=\infty$, where $\{P_1,\dots,P_{2g+2}\}$ denotes the set of the branching points.  $\mathcal{M}$ is the compact Riemann surface of the algebraic curve $w^2=\prod_{i=1}^{2g+1}(z-x_i)$. We represent $\mathcal{M}$ in Figure 1; with $\mathcal{B}_1:=\{a_1,\dots,a_g,b_1,\dots,b_g \}$ we denote an explicit choice of canonical homology basis for $H_1(\mathcal{M})$ and with $\{\omega_i\}_{i=1,\dots,g}$ the dual basis of holomorphic abelian differentials. The $n=2g+2$ sided polygon $P_n$ is represented in Figure 2; the restriction $z\mid_{P_n}$ is a homeomorphism between $P_n$ and $\mathbb{H}^+ \subset \mathbb{C}_{\infty}$.
 Let $\varphi$ be the Abel-Jacobi map for $\mathcal{M}$ \cite{FK}, i.e.  $\varphi:\mathcal{M}\rightarrow \mathcal{J}(\mathcal{M})$, $ \varphi(P):= \int_{P_1}^{P}\underline{\omega}$: we explicitly choose the branching point $P_1$ as base point for $\varphi$.  The main result of Section 6 is 

\begin{Theorem}
The integral kernels for the relevant superpropagators $G_{\mathcal{S}_i}$   in presence of $n \geq 4$ branes with  $n=2g+2$,   are given by 

\small{
\begin{eqnarray*}
\theta(Q,P)_{\mathcal{S}_1}=\frac{1}{2\pi}d\arg \frac{ \vartheta(\varphi(P)-\varphi(Q)+\mathcal{A}_g,\Omega) \vartheta(\overline{\varphi(P)}-\varphi(Q)+ \bar{\mathcal{A}}_g,\Omega)}{\vartheta(\varphi(P)+\varphi(Q)+\bar{\mathcal{A}}_g,\Omega) \vartheta(\overline{\varphi(P)}+\varphi(Q)+\mathcal{A}_g,\Omega)}-4\mathcal{Z}_{\mathcal{S}_1}(Q,P) & \\
\theta(Q,P)_{\mathcal{S}_2}=\frac{1}{2\pi} d\arg\frac{ \vartheta(\varphi(P)-\varphi(Q)+\mathcal{A}_g,\Omega) \vartheta(\overline{\varphi(P)}+\varphi(Q)+ \bar{\mathcal{A}}_g,\Omega)}{\vartheta(\varphi(P)+\varphi(Q)+\bar{\mathcal{A}}_g,\Omega) \vartheta(\overline{\varphi(P)}-\varphi(Q)+\mathcal{A}_g,\Omega)}-4\mathcal{Z}_{\mathcal{S}_2}(Q,P)& 
 \end{eqnarray*}} \normalsize{
where $\mathcal{Z}_{\mathcal{S}_i}(Q,P):=\left\{ \begin{array}{c} \im \varphi_i(Q) (\im \Omega)^{-1}d\re  \varphi_j(P)~~~ i=1 \\ \im \varphi_i(P) (\im \Omega)^{-1}d\re  \varphi_j(Q)~~~ i=2\end{array}\right. $,  $\varphi: \mathcal{M} \rightarrow \mathcal{J}(\mathcal{M}):=\mathbb{C}^g/\null^tnI+\null^tm\Omega$ is the Abel-Jacobi map for the hyperelliptic Riemann surface $\mathcal{M}$ of genus $g$ which realizes the two sheeted branched covering $z: \mathcal{M} \rightarrow \mathbb{C}_{\infty}$ with 
 branching points $\{ P_1,\dots,P_{2g+2}\}$ such that $z(P_i) \in \mathbb{R} \cup \{\infty\}$, $i=1,\dots,2g+2$, $(P,Q) \in P_n \times P_n$ and $d=d_P+d_Q$. $\mathcal{A}_g$ is  any  non singular odd half period in $\mathcal{J}(\mathcal{M})$ of the form $\mathcal{A}_g=\varphi(P_3P_5\dots\widehat{P}_{2j+1}\dots P_{2g+1})+\mathcal{K}$
 $j=1,\dots,g$, with  $\mathcal{K}$   vector of Riemann constants and $~\widehat{\null}$ means omission.
 Moreover,  given the compact notation $u =\varphi(P), v= \varphi(Q)$, $\forall (P,Q) \in P_n \times P_n$, we have the additional boundary conditions

\begin{eqnarray*}
&&\theta(v,u)_{\mathcal{S}_1}=\theta(-\bar{v},u)_{\mathcal{S}_1}=\theta(v,\bar{u})_{\mathcal{S}_1}=\theta(e^{(1)}-\bar{v},u)_{\mathcal{S}_1}= \nonumber 
\\ \label{AMM}
&&\theta(e^{(1)}+\dots +e^{(j)}-\bar{v},u)_{\mathcal{S}_1}=\theta(v,\tau^{(1)}+\bar{u})_{\mathcal{S}_1}=\theta(v,\tau^{(1)}+\tau^{(j)}+\bar{u})_{\mathcal{S}_1},
\\ 
&&\theta(v,u)_{\mathcal{S}_2}=\theta(\bar{v},u)_{\mathcal{S}_2}=\theta(v,-\bar{u})_{\mathcal{S}_2}=\theta(\tau^{(1)}+\bar{v},u)_{\mathcal{S}_2}= \nonumber 
\\ \label{ANN}
&&\theta(\tau^{(1)}+ \tau^{(j)}+\bar{v},u)_{\mathcal{S}_2}=\theta(v,e^{(1)}-\bar{u})_{\mathcal{S}_2}=\theta(v,e^{(1)}+\dots+e^{(j)}-\bar{u})_{\mathcal{S}_2}
\end{eqnarray*}
where  $j=2,\dots,g$ and  $e^{(j)}, \tau^{(j)}$   are the $j$-th columns of the  identity matrix $I$ and the  matrix of periods $\Omega$, respectively.
 By construction $\theta_{\mathcal{S}_1}(v,u)=\theta_{\mathcal{S}_2}(u,v)$, $\theta_{\mathcal{S}_2}(v,u)=\theta_{\mathcal{S}_1}(u,v)$.}
\end{Theorem}

The rest of the paper is devoted to the proof of Theorem 4. At first we study the mirror maps in the canonical setting; then we  introduce first order Riemann theta functions on $\mathcal{J}(\mathcal{M})$; we discuss then the zero set of the mirror maps and the emersion of additional contributions called zero modes: Lemma 1 gives non trivial cohomology for $(\mathcal{H}_i^n,d)$, $n \geq 4$.

\subsection{Mirror maps and reflections}Let $n$ be the number of branes. In the above theorem $n$ is even; in fact if the number of branes $n$ is odd with $n \geq 5 $, then  we can recover a  $n-1$ branes case: for this reason it is sufficient to develop  the analysis in presence of $n=2g+2$ branes. More explicitly, in presence of $n \geq 5$ branes, with  $n$ odd, we should identify a polygon $P_n$ on a  hyperelliptic curve  given by $y^2=\prod^{n-1}_{i=1}(z-x_i)$ with real  $x_i$. But we can impose, for example,  the transformation $z \rightarrow \frac{z-x_j}{z-x_k}$, $k$ even, to get a hyperelliptic curve of the type $y'^2=\prod^{n-2}_{i=1}(z'-y_i)$: the analysis is then reduced to an even number of branes case.
 The choice of the homology basis $\mathcal{B}_1$  for $\mathcal{M}$ gives

\begin{eqnarray}
\varphi(P_1)&=&0,  \nonumber \\
\varphi(P_2)&=&\frac{1}{2}e^{(1)}, \nonumber  \\
\vdots \nonumber \\
\varphi(P_{2k+1})&=&\frac{1}{2}(e^{(1)}+\dots +e^{(k)}+\tau^{(1)}+\tau^{(k+1)}), \nonumber \\
\varphi(P_{2k+2})&=&\frac{1}{2}(e^{(1)}+\dots +e^{(k+1)}+\tau^{(1)}+\tau^{(k+1)}), \nonumber\\
\vdots \nonumber \\
\varphi(P_{2g+1})&=&\frac{1}{2}(e^{(1)}+\dots +e^{(g)}+\tau^{(1)}), \nonumber \\
\varphi(P_{2g+2})&=&\frac{1}{2}\tau^{(1)}. \nonumber 
\end{eqnarray}  
for all $k=1,\dots g-1$. Let

 \begin{eqnarray}
 u:=\varphi\mid_{P_n};
\label{Abel}
 \end{eqnarray}
 in the sequel we will use the notation $u=u(P), v=u(Q)$, $\forall (P,Q) \in P_n \times P_n$.
 
  As $\mathcal{M}$ is hyperelliptic we write the elements of the  basis of holomorphic abelian differentials as $\omega_i(z):=\mathbb{I}_{ij}\frac{z^{j-1}}{w(z)}dz$; from the definition of $\mathcal{B}_1$ it follows that  the coefficients $\mathbb{I}_{ij}$ are real and the elements $\tau_{ij}=(\Omega)_{ij}$ are purely imaginary $\forall i,j \in \{1,\dots,g\}$.  We rewrite (\ref{Abel}) as $u(P)=(\mathbb{I}_{1j} \mathcal{A}_{P_1,P}^j,\dots,\mathbb{I}_{gj} \mathcal{A}_{P_1,P}^{j})$  with $\mathcal{A}_{P_1,P}^j := \int_{x_1}^{z(P)}\frac{s^{j-1}}{w(s)}ds$, $ w(s)=\sqrt{\prod_{k=1}^{2g+1}(s-x_k)}, i=1,\dots,g$. By definition, for every point $P \in \mathcal{M}-\{P_1,P_2,\dots,P_{2g+2} \}$  there exists a point $Q \in \mathcal{M}$, $Q \neq P$ such that $z(P)=z(Q)$: the pair $(P,Q)$ projects to the same point $z(P)$ on the Riemann sphere $\mathbb{C}_{\infty}$ with $P,Q$ belonging to different sheets. This means that  selecting $P_n \subset \mathcal{M}$ then  (\ref{Abel}) is well defined.

\begin{figure}[h!]
  \begin{center}
  \epsfig{figure=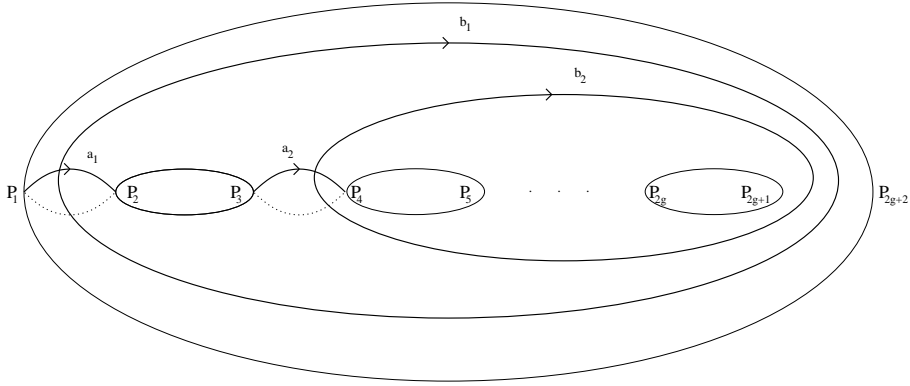,width=12cm}
  \end{center}
  \caption{Homology Basis  $\mathcal{B}_1$ for the hyperelliptic curve $\mathcal{M}$}
  \label{figure54}
  \end{figure}

  The boundary conditions for the integral kernels $\theta(P,Q)_{\mathcal{S}_i}$ are given by (\ref{bc}); we project homeomorphically each point of $P_n$ to the Jacobian variety $\mathcal{J}(\mathcal{M})$  via (\ref{Abel}) and we introduce there  the $\psi(u,v)_{\mathcal{S}_i}$ maps. This allows to express explicitly the boundary conditions (\ref{bc}) for the integral kernels in terms of the variables $(u,v)$ in the Jacobian variety itself.

  We consider now the reflections for  $\psi(u,v)_{\mathcal{S}_1}$. We define the intervals $C_k:=[x_{k-1},x_{k}]$, $k=1,\dots,2g+2$, where $x_k=z(P_k)$ for $k=1,\dots,2g+1$  and, as usual $x_0=-\infty$, $x_{2g+2}=+\infty$, as $z(P_{2g+2})=\infty$.
If $P\in C_1$ then $u(P) \in i\mathbb{R}^g$ and the reflection is $\psi(u,v)_{\mathcal{S}_1}=-\psi(-\bar{u},v)_{\mathcal{S}_1}$.
Selecting $Q \in C_2$ we get  $v(Q) \in
 \mathbb{R}^g $; this implies $\psi(u,v)_{\mathcal{S}_1}=-\psi(u,\bar{v})_{\mathcal{S}_1}$.
 With $P \in C_3$ we get $u(P)=\frac{1}{2}e^{(1)}+i\mathbb{R}^g$ and so $\psi(u,v)_{\mathcal{S}_1}=-\psi(e^{(1)}-\bar{u},v)_{\mathcal{S}_1}$. Moreover if $Q \in C_4$ then $v(Q) = \frac{1}{2}(e^{(1)}+\tau^{(1)}+\tau^{(2)})+\mathbb{R}^g$ and $\psi(u,v)_{\mathcal{S}_1}=-\psi(u,\tau^{(1)}+\tau^{(2)}+\bar{v})_{\mathcal{S}_1}$. In presence of  odd branes (different from $C_{3}$ and $C_1$), i.e. for $P \in C_{2j+1}$, $\forall j=2,\dots,g$, we have  $ u(P)=  \frac{1}{2}(e^{(1)}+\dots +e^{(j)}+\tau^{(1)}+\tau^{(j)})+ i\mathbb{R}^g$   and the reflections $\psi(u,v)_{\mathcal{S}_1}=-\psi(e^{(1)}+\dots+e^{(j)}-\bar{u},v)_{\mathcal{S}_1}$. With even branes (different from $C_{2}$ and $C_{2g+2}$ ) that is for $Q \in C_{2j}$, $\forall j=2,\dots,g$, we get $v(Q)= \frac{1}{2}(e^{(1)}+\dots +e^{(j-1)}+\tau^{(1)}+\tau^{(j)})+\mathbb{R}^g$ and the reflections $\psi(u,v)_{\mathcal{S}_1}=-\psi(u,\tau^{(1)}+\tau^{(j)}+\bar{v})_{\mathcal{S}_1}$. Analogously for $Q \in C_{2g+2}$ we have $\psi(u,v)_{\mathcal{S}_1}=-\psi(u,\tau^{(1)}+\bar{v})_{\mathcal{S}_1}$.

 Direct calculations  show that in order to compute the reflections for  $\psi(u,v)_{\mathcal{S}_2}$ we can simply consider those for $\psi(u,v)_{\mathcal{S}_1}$, then formally exchange $u$ and $v$ in $\psi(u,v)_{\mathcal{S}_1}$ and substitute the subscript $\mathcal{S}_1$ with $\mathcal{S}_2$. For example $\psi(u,v)_{\mathcal{S}_1}=-\psi(-\bar{u},v)_{\mathcal{S}_1}$ becomes $\psi(u,v)_{\mathcal{S}_2}=-\psi(u,-\bar{v})_{\mathcal{S}_2}$; $\psi(u,v)_{\mathcal{S}_1}=-\psi(u,\tau^{(1)}+\tau^{(j)}+\bar{v})_{\mathcal{S}_1}$ goes to $\psi(u,v)_{\mathcal{S}_2}=-\psi(\tau^{(1)}+\tau^{(j)}+\bar{u},v)_{\mathcal{S}_2}$ and so on. We summarize all the reflection properties for $\psi(u,v)_{\mathcal{S}_1}$ and $\psi(u,v)_{\mathcal{S}_2}$.

\begin{eqnarray*}
&&\psi(u,v)_{\mathcal{S}_1}=-\psi(-\bar{u},v)_{\mathcal{S}_1}=-\psi(u,\bar{v})_{\mathcal{S}_1}= \nonumber\\
&&-\psi(e^{(1)}-\bar{u},v)_{\mathcal{S}_1}=-\psi(u,\tau^{(1)}+\bar{v})_{\mathcal{S}_1}=-\psi(u,\tau^{(1)}+\tau^{(j)}+\bar{v})_{\mathcal{S}_1}=\nonumber \\
&&-\psi(e^{(1)}+\dots+e^{(j)}-\bar{u},v)_{\mathcal{S}_1}~ \forall j=2,\dots,g
\label{123}
\end{eqnarray*}

\begin{eqnarray*}
&&\psi(u,v)_{\mathcal{S}_2}=-\psi(u,-\bar{v})_{\mathcal{S}_2}=-\psi(\bar{u},v)_{\mathcal{S}_2}= \nonumber\\
&&-\psi(u,e^{(1)}-\bar{v})_{\mathcal{S}_2}=-\psi(\tau^{(1)}+\bar{u},v)_{\mathcal{S}_2}=-\psi(\tau^{(1)}+\tau^{(j)}+\bar{u},v)_{\mathcal{S}_2}=\nonumber \\
&&-\psi(u,e^{(1)}+\dots+e^{(j)}-\bar{v})_{\mathcal{S}_2}~ \forall j=2,\dots,g
\label{1234}
\end{eqnarray*}
Given all the other possible choices of index sets  we can reduce to a lower number of branes cases by pinching the sides  $\partial P_n^i$ of $P_n$ with the same boundary conditions imposed on $\partial P_n^{i+1}$; one repeats the process till a no more reducible case: then the computation of the integral kernels begins specifying the correct homeomorphism $u$.

\begin{figure}[h!]
  \begin{center}
  \epsfig{figure=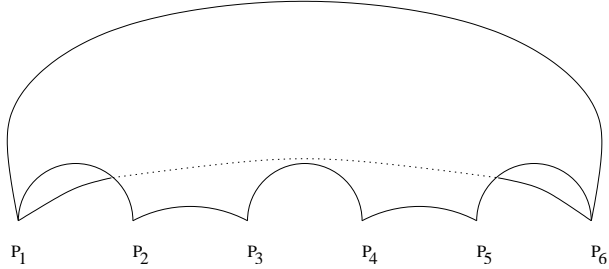,width=8cm}
  \end{center}
  \caption{Polygon $P_n$: $n=6$, or $g=2$ case}
  \label{figure5456}
  \end{figure}

 We give the explicit form of the mirror maps in terms of first order  Riemann theta functions with odd characteristics (canonical setting).
Increasing the genus $g$ of the hyperelliptic surface get more reflections to satisfy; precisely in presence of $n=2g+2$ branes we have $2g+2$ reflections to impose but it is   possible to construct mirror maps with a combination of four suitable first order Riemann theta functions with odd characteristics independently on the number of reflections. We introduce then odd non integer characteristics $(\epsilon,\epsilon')$ through the odd half periods $\mathcal{A}_g, \mathcal{B}_g, \mathcal{C}_g, \mathcal{D}_g$, where $\mathcal{A}_g = \frac{1}{2}\null^t\epsilon_1'I+\frac{1}{2}\null^t\epsilon_1\Omega$, $\mathcal{B}_g = \frac{1}{2}\null^t\epsilon_2'I+\frac{1}{2}\null^t\epsilon_2\Omega$, $\mathcal{C}_g = \frac{1}{2}\null^t\epsilon_3'I+\frac{1}{2}\null^t\epsilon_3\Omega$, $\mathcal{D}_g = \frac{1}{2}\null^t\epsilon_4'I+\frac{1}{2}\null^t\epsilon_4\Omega$ with $\epsilon_i', \epsilon_i \in \mathbb{Z}^g$ $\null^t\epsilon_i'\epsilon_i=1$ mod   2, $ i=1,\dots,4$. The basic result is the following:

\begin{Prop}
The map $\psi(u,v)_{\mathcal{S}_2} := \arg \frac{ \vartheta(u-v+\mathcal{A}_g,\Omega) \vartheta(\bar{u}+v+\mathcal{B}_g,\Omega)}{\vartheta(u+v+\mathcal{C}_g,\Omega) \vartheta(\bar{u}-v+\mathcal{D}_g,\Omega)}$ satisfies the reflection properties:

\begin{eqnarray}
&&\psi(u,v)_{\mathcal{S}_2}=-\psi(\bar{u},v)_{\mathcal{S}_2}=-\psi(u,-\bar{v})_{\mathcal{S}_2}= \nonumber\\
&&-\psi(\tau^{(1)}+\bar{u},v)_{\mathcal{S}_2}+8\pi\re  v_1=-\psi(u,e^{(1)}-\bar{v})_{\mathcal{S}_2}=\nonumber \\
&&-\psi(\tau^{(1)}+\tau^{(j)}+\bar{u},v)_{\mathcal{S}_2}+8\pi\re  v_1+8\pi\re  v_j= \nonumber \\
&&-\psi(u,e^{(1)}+\dots +e^{(j)}-\bar{v})_{\mathcal{S}_2}~~ \forall j=2,\dots,g \label{gh}
\end{eqnarray} 
 if

\begin{align*}
\mathcal{A}_g &=\mathcal{D}_g & \mathcal{A}_g &=\bar{\mathcal{B}}_g &  &\mathcal{B}_g+\mathcal{C}_g-\mathcal{A}_g-\mathcal{D}_g \in i\mathbb{R}^g ~\mbox{mod}~ \null^tm\Omega \\
\mathcal{B}_g &=\mathcal{C}_g &  \bar{\mathcal{C}}_g &=\mathcal{D}_g &   &\mathcal{A}_g+\mathcal{B}_g-\mathcal{C}_g-\mathcal{D}_g \in i\mathbb{R}^g ~\mbox{mod}~ \null^tm\Omega  \\
 \end{align*}

\end{Prop}

\begin{proof}
See Appendix B.
\end{proof}

So, the most general mirror map satisfying (\ref{gh}) and written in canonical setting is given by:

\begin{eqnarray}
\label{poi}
\psi(u,v)_{\mathcal{S}_2}&=&\arg \frac{ \vartheta(u-v+\mathcal{A}_g,\Omega) \vartheta(\bar{u}+v+ \bar{\mathcal{A}}_g,\Omega)}{\vartheta(u+v+\bar{\mathcal{A}}_g,\Omega) \vartheta(\bar{u}-v+\mathcal{A}_g,\Omega)}.
\end{eqnarray}
Following the same lines of Proposition 1 we can state a similar result for the mirror map $\psi(u,v)_{\mathcal{S}_1}$, or

\begin{Prop}

The map $\psi(u,v)_{\mathcal{S}_1} := \arg \frac{ \vartheta(u-v+\mathcal{A}_g,\Omega) \vartheta(\bar{u}-v+\mathcal{D}_g,\Omega)}{\vartheta(u+v+\mathcal{C}_g,\Omega) \vartheta(\bar{u}+v+\mathcal{B}_g,\Omega)}$ satisfies the reflections

\begin{eqnarray}
 &&\psi(u,v)_{\mathcal{S}_1}=-\psi(u,\bar{v})_{\mathcal{S}_1}=-\psi(-\bar{u},v)_{\mathcal{S}_1}= \nonumber\\
&&-\psi(u,\tau^{(1)}+\bar{v})_{\mathcal{S}_1}+8\pi\re u_1=-\psi(e^{(1)}-\bar{u},v)_{\mathcal{S}_1}=\nonumber \\
&&-\psi(u,\tau^{(1)}+\tau^{(j)}+\bar{v})_{\mathcal{S}_1}+8\pi\re u_1+8\pi\re u_j=\nonumber \\
&&-\psi(e^{(1)}+\dots+e^{(j)}-\bar{u},v)_{\mathcal{S}_1}=~~ \forall j=2,\dots,g  
 \label{wer}
 \end{eqnarray} 
 if

\begin{align*}
\mathcal{A}_g &=\bar{\mathcal{D}}_g,  & \bar{\mathcal{B}}_g &=\mathcal{C}_g \\
 \mathcal{A}_g &=-\mathcal{B}_g-\null^tn\Omega  & \mathcal{C}_g &=-\mathcal{D}_g+\null^tn\Omega \\
\mathcal{A}_g&+\mathcal{B}_g+\mathcal{C}_g+\mathcal{D}_g \in i\mathbb{R}^g ~\mbox{mod}~ \null^tm\Omega & \mathcal{A}_g&+\mathcal{C}_g-\mathcal{B}_g-\mathcal{D}_g \in i\mathbb{R}^g ~\mbox{mod}~ \null^tm\Omega \\
\end{align*}
 
 for any vector $n \in \mathbb{Z}^g$.
 
\end{Prop}

If we select odd characteristics through $\mathcal{A}_g,$ $ \mathcal{B}_g,$ $\mathcal{C}_g$, and $\mathcal{D}_g$ fulfilling the hypotesis of Proposition 2  we can write the most general  mirror map in canonical formalism  which satisfies (\ref{wer}), or

\begin{eqnarray}
\label{poipoi}
\psi(u,v)_{\mathcal{S}_1}&=&\arg \frac{ \vartheta(u-v+\mathcal{A}_g,\Omega) \vartheta(\bar{u}-v+ \bar{\mathcal{A}}_g,\Omega)}{\vartheta(u+v-\bar{\mathcal{A}}_g,\Omega) \vartheta(\bar{u}+v-\mathcal{A}_g,\Omega)}.
\end{eqnarray}

\subsection{Zero set of the mirror maps} We characterize now the zero set of the mirror maps. We need to introduce some definitions for divisors before.
 Given a compact Riemann surface $\mathcal{M}$ of genus $g$, then an odd half period $\mathcal{A}_g$ of its Jacobian variety $\mathcal{J}(\mathcal{M})$ is called non singular if it can be written as $\mathcal{A}_g=\varphi(\mathcal{D}_{g-1})+\mathcal{K}$ with $\mathcal{D}_{g-1}$  non special integral divisor of degree $g-1$ on $\mathcal{M}$, $\mathcal{K}$ vector of Riemann constants and $\varphi$ Abel-Jacobi map. We say that integer odd characteristics $\epsilon, \epsilon' \in \mathbb{Z}^g$ are non singular if the corresponding odd half period $\mathcal{A}_g=\frac{1}{2}\null^t\epsilon'I+\frac{1}{2}\null^t\epsilon\Omega$ is non singular. An integral divisor $D=k_iP^i$ , $P^i \in \mathcal{M}, k_i \in \mathbb{Z}$ is special if $i(D) > g-\mbox{deg}D \Longleftrightarrow r(-D) >1$ for Riemann-Roch theorem. Here $i(D)$ is the 
 index of specialty and $r(D)$ the dimension of  the divisor $D$.
  The existence of non singular odd half periods is a result shown, for example,  in \cite{FK}. Extending the Abel-Jacobi map $\varphi$ (\ref{Abel}) to arbitrary divisors $D=k_jP^j$ with $\varphi(D):= k_i \varphi(P^i)$, we recall that, with   $D_1$, $D_2$ integral non special divisors of degree $g$ such that  $\varphi(D_1)=\varphi(D_2)$, then $D_1 =D_2$. This result comes from the fact that Abel's theorem \cite{FK} implies that $D_1-D_2=(f)$, where $(f)$ is the divisor of a meromorphic function $f$ on $\mathcal{M}$; by non specialty it follows  $r(-D_1)=r(-D_2)=1$, that is $D_1=D_2$.  
Let us consider the multivalued map on  $\mathcal{M}$: $P \rightarrow \vartheta(\varphi(P)-e)$; if it is not identically null, then it has $g$ zeros,  and the zero divisor $\mathcal{Z}_g$ satisfies $\varphi(\mathcal{Z}_g)+\mathcal{K}=e$. For more details see Appendix B. Moreover it can be shown that $P \rightarrow \vartheta(\varphi(P)-e)$ vanishes identically on $\mathcal{M}$ iff $e=\varphi(\mathcal{D}_{g})+\mathcal{K}$ with $\mathcal{D}_g$ integral special divisor of degree $g$. 

By writing $\psi(u,v)_{\mathcal{S}_1} := \arg \chi_{\mathcal{S}_1}(u,v)$, $\psi(u,v)_{\mathcal{S}_2} := \arg \chi_{\mathcal{S}_2}(u,v)$ where $\psi(u,v)_{\mathcal{S}_1}$ and $\psi(u,v)_{\mathcal{S}_2}$ are given by  (\ref{poipoi},\ref{poi}) we  can state the following

\begin{Prop}
Let $\mathcal{M}$ be the  hyperelliptic Riemann surface of genus $g$ of Section 6.1 and let the mirror maps $\psi(u,v)_{\mathcal{S}_1}$, $\psi(u,v)_{\mathcal{S}_2}$ be given by (\ref{poipoi},\ref{poi}) with $\mathcal{A}_g$ non singular odd half period on the Jacobian variety of $\mathcal{M}$. Choose 
 
\begin{equation}
\label{AAA}
\mathcal{A}_g=\varphi(P_3P_5\dots\widehat{P}_{2j+1}\dots P_{2g+1})+\mathcal{K}:=\varphi(\mathcal{D}_{g-1,j})+\mathcal{K}
\end{equation}
  with $j=1,\dots,g$, where $\{P_{2k+1}\}_{k=1,\dots,g}$ are the odd branching points of $z: \mathcal{M}\rightarrow \mathbb{C}_{\infty}$, $\mathcal{K}$ is the vector of Riemann constants (see Appendix B) and caret means omission. Then the zero divisor of 
$\chi_{\mathcal{S}_1}(u,v)$ and $\chi_{\mathcal{S}_2}(u,v)$ on $P_n^{\circ}\times P_n^{\circ} \subset \mathcal{M} \times \mathcal{M}$  consists only of the  diagonal $\Delta=\left\{ (P,P)\mid  P \in P_n^{\circ} \right\}$.
\end{Prop}

\begin{proof}

First of all we note that $\mathcal{A}_g$ can be  written as $\mathcal{A}_g=\varphi(P_{2j+1}) $ mod $\null^tnI+\null^tm\Omega$, $j=1,\dots,g$ as the vector of Riemann constants is given by $\mathcal{K}=\sum_{j=1}^{g}\varphi(P_{2j+1})$, for the hyperelliptic curve $\mathcal{M}$ of Section 6.1 once we select the canonical homology basis $\mathcal{B}_1$. In this setting it can be shown \cite{FK} that $i(P_3P_5\dots P_{2g+1})=0$; this implies that the divisor $P_3P_5\dots\widehat{P}_{2j+1}\dots P_{2g+1}$ is non special and  with the choice (\ref{AAA}) $\mathcal{A}_g$ is a not singular odd half period.

In order to discuss the zero divisor of $\chi_{\mathcal{S}_1}(u,v)$ we begin by writing $\psi(u,v)_{\mathcal{S}_1} =\arg \chi_{\mathcal{S}_1}(u,v)=\arg \frac{ \vartheta(u-v+\mathcal{A}_g,\Omega) \vartheta(\bar{u}-v+ \bar{\mathcal{A}}_g,\Omega)}{\vartheta(u+v-\bar{\mathcal{A}}_g,\Omega) \vartheta(\bar{u}+v-\mathcal{A}_g,\Omega)}=\arg \frac{ \vartheta(u-v-\mathcal{A}_g,\Omega) \vartheta(u+\bar{v}-\mathcal{A}_g,\Omega)}{\vartheta(u+v-\mathcal{A}_g,\Omega) \vartheta(u-\bar{v}-\mathcal{A}_g,\Omega)}$. Now  we proceed to study the zero divisor of the multivalued function

\begin{equation*}
 P \rightarrow \chi_{\mathcal{S}_1}(u(P),v(Q))=\frac{ \vartheta(u(P)-v(Q)-\mathcal{A}_g,\Omega) \vartheta(u(P)+\overline{v(Q)}-\mathcal{A}_g,\Omega)}{\vartheta(u(P)+v(Q)-\mathcal{A}_g,\Omega) \vartheta(u(P)-\overline{v(Q)}-\mathcal{A}_g,\Omega)}
\end{equation*}
with $Q \in \mathcal{M}\backslash\{P_{i}\}_{i=1,\dots,g}$ fixed. The thetas are multivalued functions on $\mathcal{M}$ but the zero divisor is well defined as the multivaluedness generates a multiplicative non vanishing factor. As $Q \not\in \mathcal{D}_{g-1,j}$, then the divisor $Q\mathcal{D}_{g-1,j}$ is not special and the zero divisor of the multivalued holomorphic function on $\mathcal{M}$: $P \rightarrow \vartheta(u(P)-v(Q)-\mathcal{A}_g,\Omega)$ is given by $Q\mathcal{D}_{g-1,j}$.
 Analogously $P \rightarrow \vartheta(u(P)+v(Q)-\mathcal{A}_g,\Omega)$ has the zero divisor $S\mathcal{D}_{g-1,j}$ with $S=J(Q)$, that is $S$ is the image of $Q$ under the hyperelliptic involution $J$. Explicitly $J$ acts as follows: we write  $J(Q)=Q$ if $Q$ is a branching point for $\mathcal{M}$ (it is not our case), otherwise $J(Q)=S$ with $z(Q)=z(S)$ and $z:\mathcal{M}\rightarrow \mathbb{C}_{\infty}$ is the two sheeted branched covering of $\mathbb{C}_{\infty }$: for this reason $J$ is also called the "sheet exchange". From the definition it follows that $v(Q)=\varphi(Q)=-\varphi(J(Q)):=-v(S)$. We continue with $P \rightarrow  \vartheta(u(P)-\overline{v(Q)}-\mathcal{A}_g,\Omega)$; from the definition of the Abel-Jacobi map we get $\overline{v(Q)}=\overline{\int_{P_1}^Q \underline{\omega}} =
(\mathbb{I}_{1j} \int_{\overline{\gamma}} \frac{s^{j-1}ds}{\pm w(s)},\dots,\mathbb{I}_{gj}\int_{\overline{\gamma}} \frac{s^{j-1}ds}{\pm w(s)})$ 
where $\gamma$ is a smooth  curve  joining $x_1$ and $z(Q)$ and $w(s)=\sqrt{\prod_{i=1}^{2g+1}(s-x_i)}$ with $\overline{w(\bar{s})}=\pm w(s)$, for Schwarz reflection principle applied to regions in $\mathbb{H}^+$. Whenever we have the + sign, introducing the automorphism on $\mathcal{M}$
$C: Q \rightarrow C(Q)=T$, with $z(Q)=\overline{z(T)}$, we get $ \overline{v(Q)}=\overline{\varphi(Q)}=\varphi(T):=v(T)$. Whenever we have the - sign (for example, for $Q$ such that $z(Q) \in (x_{2i},x_{2i+1}), i=0,\dots,g, x_0=-\infty $) we compose $C$ with the sheet exchange $J$.
 
   So we get    the zero divisor $\left\{ \begin{array}{c} C(Q) \\ JC(Q) \end{array} \right\}\mathcal{D}_{g-1,j}$  for  $P \rightarrow  \vartheta(u(P)-\overline{v(Q)}-\mathcal{A}_g,\Omega)$. The bracket selects only one of the two divisors $C(Q)$ and $JC(Q)$,  depending on the sign of $\overline{w(\bar{s})}=\pm w(s)$, as above. All the  considerations so far imply that the multivalued function $ P \rightarrow \vartheta(u(P)+\overline{v(Q)}-\mathcal{A}_g,\Omega)$ has zero divisor given by $\left\{ \begin{array}{c} JC(Q) \\ C(Q) \end{array} \right\}\mathcal{D}_{g-1,j}$. Collecting  all the zero divisors for the thetas  we obtain the divisor of the function  $P \rightarrow \chi_{\mathcal{S}_1}(u(P),v(Q))$:

\begin{equation*}
(\chi_{\mathcal{S}_1}(u(P),v(Q))) = \frac{Q\mathcal{D}_{g-1,j}\left\{ \begin{array}{c} JC(Q) \\ C(Q) \end{array} \right\}\mathcal{D}_{g-1,j}}{J(Q)\mathcal{D}_{g-1,j}\left\{ \begin{array}{c} C(Q) \\ JC(Q) \end{array} \right\}\mathcal{D}_{g-1,j}};
\end{equation*}
 i.e. $P \rightarrow \chi_{\mathcal{S}_1}(u(P),v(Q))$ is null only for $P=Q$ on $P_n^{\circ}$ for the injectivity of the Abel-Jacobi map.

 As $\psi(u,v)_{\mathcal{S}_1} =\arg \chi_{\mathcal{S}_1}(u,v)=\arg \frac{ \vartheta(v-u-\mathcal{A}_g,\Omega) \vartheta(v-\bar{u}-\bar{\mathcal{A}}_g,\Omega)}{\vartheta(v+u-\bar{\mathcal{A}}_g,\Omega) \vartheta(v+\bar{u}-\mathcal{A}_g,\Omega)}$,  then we study the zero divisor of

\begin{equation*}
 Q \rightarrow \chi_{\mathcal{S}_1}(u(P),v(Q))=\frac{ \vartheta(v(Q)-u(P)-\mathcal{A}_g,\Omega) \vartheta(v(Q)-\overline{u(P)}-\bar{\mathcal{A}}_g,\Omega)}{\vartheta(v(Q)+u(P)-\bar{\mathcal{A}}_g,\Omega) \vartheta(v(Q)+\overline{u(P)}-\mathcal{A}_g,\Omega)}
\end{equation*}
with $P \in \mathcal{M}\backslash\{P_{i}\}_{i=1,\dots,g}$ fixed. The analysis follows the same lines here discussed and $Q \rightarrow \chi_{\mathcal{B}^1,\mathcal{S}_1}(u(P),v(Q))$ is null only for $P=Q$ on $P_n^{\circ}$ for the injectivity of the Abel-Jacobi map. 
Also the multivalued map  $\chi_{\mathcal{S}_2}(u,v)$ presents the same behaviour.

\end{proof}

\subsection{Abel-Jacobi map and zero modes for the superpropagators}

In presence of more than three branes we have the emersion of zero modes contributions as the cohomology of $(\mathcal{H}^n_{\mathcal{S}_i})$ is non trivial (Lemma 1). The contributions denoted with $\mathcal{Z}_{\mathcal{S}_i}(Q,P)$ in the definition of the integral kernels (Def.1) must absorb the extra terms in the reflections of the mirror maps   to get the correct boundary conditions for the $\theta(Q,P)_{\mathcal{S}_i}$. In the elliptic case ($g=1$) the image of $P_4$ in $\mathcal{J}(\mathcal{M})$ is given by the set $u(P_4)=\{u \in \mathbb{C}/nI+m\tau  \mid \re  u \in [0,\frac{1}{2}], ~\im  u \in [0,\frac{t}{2}]\}$, where $t=\im \tau >0$. We know that dim$H^1(\mathcal{H}^n_{\mathcal{S}_i})=1$;  explicit basis are given by $\rho_{\mathcal{S}_1}=d\im u(P)= \frac{1}{2i}[\omega(P)-\overline{\omega(P)}]$ and $\rho_{\mathcal{S}_2 }=d\re  u(P)=\frac{1}{2}[\omega(P)+\overline{\omega(P)}]$, respectively. With $\omega(P)$ here we denote the basis of holomorphic abelian differentials dual to the canonical homology basis $\mathcal{B}_1$. Motivated by these considerations, we write
   
\begin{Lem}
Let  $\theta(Q,P)_{\mathcal{S}_i}$ be given by (\ref{intker}) for $n=2g+2$ branes with  

\begin{eqnarray*}
\mathcal{Z}_{\mathcal{S}_1}(Q,P)&=&\im v_i\mathcal{O}^{ik}d\re  u_k, \nonumber \\
\mathcal{Z}_{\mathcal{S}_2}(Q,P)&=&\im u_i\mathcal{O}^{ik}d\re  v_k,
\end{eqnarray*}
$i,k=1,\dots,g$; then $\theta(Q,P)_{\mathcal{S}_i}$ satisfies  (\ref{bc}) for $\mathcal{O}=4 (\Im m\Omega)^{-1}$.

\begin{proof} 
Apply the boundary conditions (\ref{bc}) to $\mathcal{Z}_{\mathcal{S}_i}$; the additional contributions generated by the reflections cancel with those of the mirror maps, as in Prop.1-2.
\end{proof}

\end{Lem}

\subsection{Superpropagators with $n \geq 4$ branes: explicit formulas} We are ready to write down the explicit formulas for the superpropagators in the $n \geq 4$ branes case. We have found the generalized angle functions through theta functions, we have studied thier zero divisor and we have introduced the zero modes terms; they correct the additional contributions generated in the reflections of the mirror maps. Collecting all these results we get the superpropagators ($P,Q \in P_n \times P_n$)

 \begin{eqnarray*}
\label{propag1}
\theta(Q,P)_{\mathcal{S}_1}=\frac{1}{2\pi}d\arg \frac{ \vartheta(\varphi(P)-\varphi(Q)+\mathcal{A}_g,\Omega) \vartheta(\overline{\varphi(P)}-\varphi(Q)+ \bar{\mathcal{A}}_g,\Omega)}{\vartheta(\varphi(P)+\varphi(Q)+\bar{\mathcal{A}}_g,\Omega) \vartheta(\overline{\varphi(P)}+\varphi(Q)+\mathcal{A}_g,\Omega)}-4\mathcal{Z}_{\mathcal{S}_1}(Q,P) & \\
\label{propag2}
\theta(Q,P)_{\mathcal{S}_2}=\frac{1}{2\pi} d\arg\frac{ \vartheta(\varphi(P)-\varphi(Q)+\mathcal{A}_g,\Omega) \vartheta(\overline{\varphi(P)}+\varphi(Q)+ \bar{\mathcal{A}}_g,\Omega)}{\vartheta(\varphi(P)+\varphi(Q)+\bar{\mathcal{A}}_g,\Omega) \vartheta(\overline{\varphi(P)}-\varphi(Q)+\mathcal{A}_g,\Omega)}-4\mathcal{Z}_{\mathcal{S}_2}(Q,P)& 
 \end{eqnarray*}
where $\mathcal{Z}_{\mathcal{S}_i}(Q,P):=\left\{ \begin{array}{c} \im \varphi_i(Q) (\im \Omega)^{-1}d\re  \varphi_j(P)~~~ i=1 \\ \im \varphi_i(P) (\im \Omega)^{-1}d\re  \varphi_j(Q)~~~ i=2\end{array}\right.$. We have simply used (\ref{Abel}) to replace $(u,v)$ with $\varphi$. They fulfill the correct boundary conditions expressed by the index sets $\mathcal{S}_i$ with no additional terms.
Moreover it follows   $\theta(Q,P)_{\mathcal{S}_1}=\theta(P,Q)_{\mathcal{S}_2}$, $\theta(Q,P)_{\mathcal{S}_2}=\theta(P,Q)_{\mathcal{S}_1}$; the integral kernels are independent on the choice of odd non singular $\mathcal{A}_g$ as in (\ref{AAA}) and the computation of the additional boundary conditions satisfied by $\theta(Q,P)_{\mathcal{S}_i}$ is straightforward: this ends the proof of Th.4.

The integral kernels for the $n \geq 4$ branes cases present a "similarity"  with the expression of the Green function for the Laplacian operator (that is second order differential operator) on compact Riemann surfaces  in \cite{Verlinde:1986kw,Verlinde:1987sd}; such Green function it is a sum of a main part involving the prime form defined on the compact curve of genus $g$ and $g$ zero modes contributions; here instead of a single prime form we have to use the product of four first order odd Riemann theta functions to fulfill the boundary conditions. A mathematical formulation of this remark, involving Schiffer and Bergmann kernels on $\mathcal{M}$ will be given in  \cite{next}.

 \section{Conclusions}

We have written the explicit formulas for the superpropagators  of the P$\sigma$M in presence of branes. With two or three branes we used a Schwarz-Christoffel mapping to produce the integral kernels with the correct boundary conditions.  With more than three branes we have introduced hyperelliptic curves $\mathcal{M}$ of genus $g$ and first order Riemann theta functions with odd characteristics defined on the Jacobian variety $\mathcal{J}(\mathcal{M})$. The superpropagators include zero modes contributions involving the inverse of the matrix of periods for $\mathcal{M}$.   With these formulas it is possible to  study the algebraic properties and the deformation of the associative product of the algebra of observables  $\mathcal{A}$ for the P$\sigma$M in presence of branes. This should give a generalization of the $P_{\infty}$ structure on $\mathcal{A}=\Gamma(\wedge NC)$ for a single brane $C$ described in \cite{Cattaneo:2005zz},\cite{C2007}, while the non perturbative analysis necessarily leads to the definition of a  Fukaya $A_{\infty}$ category for the P$\sigma$M with branes. 
Before studying a "global"   Fukaya category inspired by the nonperturbative P$\sigma$M with branes, it is possible to analyse a local version deduced by the tools of Homological Perturbation Theory applied to a suitable differential graded category. The superpropagators  here deduced play the role of homotopy operators on the space of morphisms of such category.
As the two and three branes cases induce bimodules and morphisms of bimodules \cite{CF} \cite{Cattaneo:2005zz}  one could expect $A_{\infty}$ bimodules and morphisms for the hyperelliptic cases (or even a more general structure); moreover with linear Poisson structure  (and branes as affine subspaces) the expression of the superpropagators here obtained   gives explicit higher order formulas for the diagrammatics  developed in \cite{CT}.

\appendix

\section{Theta functions and Riemann surfaces}

We introduce briefly first order Riemann theta functions; the material here presented is standard: the reader interested in  proofs and wider expositions should consult, for example, \cite{FK}. We write the translation properties for the thetas we used in Section 6.

\begin{Def}
Let $\mathcal{G}_g$ denote the Siegel upper half spaces of genus $g$, that is the space of complex symmetric $g \times g$ matrices with positive imaginary part. We define Riemann's theta function by
\begin{equation*}
\vartheta(z,\Omega):= \Sigma_{N \in \mathbb{Z}^g} e^{2\pi i(\frac{1}{2}\null^tN\Omega N+\null^tNz)}
\end{equation*}
\end{Def} 
where $z \in \mathbb{C}^g$ (viewed as a column vector) and $\Omega \in \mathcal{G}_g$ and the sum extends over all integer vectors in $\mathbb{C}^g$. The function converges absolutely and uniformly on compact subsets of $\mathbb{C}^g \times \mathcal{G}_g$. In the sequel we will fix the matrix $\Omega$ and we will consider $\vartheta$ as a holomorphic function on $\mathbb{C}^g$.

\begin{Prop}
Let $\mu, \mu' \in \mathbb{Z}^g$. Then

\begin{equation*}
\vartheta(z+\null^t\mu'I+\null^t\mu\Omega,\Omega)=e^{2\pi i(-\null^t\mu z-\frac{1}{2}\null^t\mu\Omega\mu)}\vartheta(z,\Omega),
\end{equation*}
$ \forall z \in \mathbb{C}^g, \Omega \in \mathcal{G}_g$ and with $I$  the  $g\times g$ identity matrix.
\end{Prop}

In particular

\begin{eqnarray}
\vartheta(z+e^{(j)},\Omega)&=&\vartheta(z,\Omega), \nonumber \\
\vartheta(z+\tau^{(j)},\Omega)&=&e^{2\pi i(-z_k-\frac{\tau_{kk}}{2})}\vartheta(z,\Omega), \nonumber \\
\label{translat}
\vartheta(-z,\Omega)&=&\vartheta(z,\Omega),
\end{eqnarray} 
 where $e^{(j)}$ is the $j$-th column of the identity matrix $I$ and $\tau^{(j)}$ is the
 $j$-th column of $\Omega$. We continue with
\begin{Def}
Let $\epsilon$, $\epsilon' \in \mathbb{R}^{g}$; the holomorphic functions on $\mathbb{C}^g \times \mathcal{G}_g$
\begin{equation*}
\vartheta\left[\begin{array}{c} \epsilon \\ \epsilon' \end{array}\right]
(z,\Omega):= \Sigma_{N \in \mathbb{Z}^g}e^{2\pi i(\frac{1}{2}\null^t(N+\frac{\epsilon}{2})\Omega (N+\frac{\epsilon}{2})+\null^t(N+\frac{\epsilon}{2})(z+\frac{1}{2}))}
\end{equation*}
are called (first order) theta functions with characteristics.
\end{Def} 

When selecting $\epsilon$, $\epsilon' \in \mathbb{Z}^g$ (integer characteristics) it is easy to prove that $\vartheta\left[\begin{array}{c} \epsilon \\ \epsilon' \end{array}\right]
(-z,\Omega)=e^{2\pi i (\frac{\null^{t}\epsilon\epsilon'}{2})}\vartheta\left[\begin{array}{c} \epsilon \\ \epsilon' \end{array}\right]
(z,\Omega)$, i.e. theta functions with integer characteristics are odd if $\null^{t}\epsilon'\epsilon=1$ mod 2, even when $\null^{t}\epsilon^{'}\epsilon=0$ mod 2. Moreover it follows $\vartheta\left[\begin{array}{c} \epsilon \\ \epsilon' \end{array}\right]
(z,\Omega)=e^{2\pi i ( \frac{1}{8}\null^{t}\epsilon \Omega \epsilon +\frac{1}{2}\null^{t}\epsilon z + \frac{1}{4}\null^{t}\epsilon\epsilon' )         }\vartheta(z+ \frac{\null^t\epsilon'}{2}I+ \frac{\null^t\epsilon}{2}\Omega,\Omega)$, i.e.  if $\vartheta\left[\begin{array}{c} \epsilon \\ \epsilon' \end{array}\right] (z,\Omega)$ is odd then $\vartheta( \frac{\null^t\epsilon'}{2}I+ \frac{\null^t\epsilon}{2}\Omega,\Omega)=0$; the Riemann theta function  vanishes at the odd points of order two in the lattice generated by the columns of $I$ and $\Omega$: they are  called odd half periods. We introduce now $\vartheta\left[\begin{array}{c} \epsilon \\ \epsilon' \end{array}\right] \circ \varphi$, where $\varphi$ is the  Abel-Jacobi map for $\mathcal{M}$, compact Riemann surface of genus $g$. $\vartheta\left[\begin{array}{c} \epsilon \\ \epsilon' \end{array}\right] \circ \varphi$  is multivalued but its zeroes are well defined on $\mathcal{M}$ becouse the multivaluedness is multiplicative with a non vanishing factor. Then we have the classical

\begin{Theorem} 
Let $\mathcal{M}$ be a compact Riemann surface of genus $g \geq$ 1 with canonical homology basis $\{a_i,b_i  \}_{\{i=1,\dots,g \} }$. Let $\vartheta\left[\begin{array}{c} \epsilon \\ \epsilon' \end{array}\right] \circ \varphi$ be the first order Riemann theta function associated  with $( \mathcal{M}, \{a_i,b_i  \}_{\{i=1,\dots,g \} }) $ and let $\varphi$ be the Abel-Jacobi map for $ \mathcal{M}$. Then $\vartheta\left[\begin{array}{c} \epsilon \\ \epsilon' \end{array}\right] \circ \varphi$ is either identically zero as a function on $\mathcal{M}$ or else it has precisely $g$ zeros on $\mathcal{M}$. In this case let $P_1P_2\dots P_{g}$ be the
divisor of zeros. We then have $\varphi(P_1P_2\dots P_{g})=-\frac{\Omega\epsilon}{2}-\frac{I\epsilon'}{2}-\mathcal{K}$, where $\mathcal{K}$ is the  vector of Riemann constants which depends on the canonical homology basis and the base point for the Abel-Jacobi map.
\end{Theorem}

 We remind that   for every compact Riemann surface of genus  $g \geq 1$ $\mathcal{K}$ is a half period
 in the Jacobian variety of $\mathcal{M}$; in particular, selecting the hyperelliptic curve $\mathcal{M}$ described in Section 6.1 and the canonical homology basis $\mathcal{B}_1$, it is not hard to show that $\mathcal{K}=\sum_{j=1}^{g}\varphi(P_{2j+1})$, where $\{ P_{2j+1} \}_{j=1,\dots,g}$ is the set of odd branching points of $\mathcal{M}$, i.e.
 
 \begin{equation*}
  \mathcal{K}=\frac{1}{2}(ge^{(1)}+(g-1) e^{(2)}+\dots+e^{(g)}+g\tau^{(1)}+\tau^{(2)}+\dots+\tau^{(g)}) ~\mbox{mod}~\null^tnI+\null^tm\Omega
 \end{equation*}

\section{Proof of Proposition 1}

We write here the proof of Proposition 1 of Section 6.1; it states that the map $\psi(u,v)_{\mathcal{S}_2} := \arg \frac{ \vartheta(u-v+\mathcal{A}_g,\Omega) \vartheta(\bar{u}+v+\mathcal{B}_g,\Omega)}{\vartheta(u+v+\mathcal{C}_g,\Omega) \vartheta(\bar{u}-v+\mathcal{D}_g,\Omega)}$ satisfies the reflection properties:

\begin{eqnarray}
&&\psi(u,v)_{\mathcal{S}_2}=-\psi(\bar{u},v)_{\mathcal{S}_2}=-\psi(u,-\bar{v})_{\mathcal{S}_2}= \nonumber\\
&&-\psi(\tau^{(1)}+\bar{u},v)_{\mathcal{S}_2}+8\pi\re  v_1=-\psi(u,e^{(1)}-\bar{v})_{\mathcal{S}_2}=\nonumber \\
&&-\psi(\tau^{(1)}+\tau^{(j)}+\bar{u},v)_{\mathcal{S}_2}+8\pi\re  v_1+8\pi\re  v_j=  \nonumber \\ 
\label{ABCD}
&&-\psi(u,e^{(1)}+\dots+e^{(j)}-\bar{v})_{\mathcal{S}_2}~~ \forall j=2,\dots,g 
\end{eqnarray} 
 provided

\begin{align*}
\mathcal{A}_g &=\mathcal{D}_g & \mathcal{A}_g &=\bar{\mathcal{B}}_g &  &\mathcal{B}_g+\mathcal{C}_g-\mathcal{A}_g-\mathcal{D}_g \in i\mathbb{R}^g  ~\mbox{mod}~\null^tm\Omega  \\
\mathcal{B}_g &=\mathcal{C}_g &  \bar{\mathcal{C}}_g &=\mathcal{D}_g &   &\mathcal{A}_g+\mathcal{B}_g-\mathcal{C}_g-\mathcal{D}_g \in i\mathbb{R}^g ~\mbox{mod}~\null^tm\Omega   \\
 \end{align*}

\begin{proof}

Imposing the reflections given by (\ref{123}) we get:

\begin{eqnarray}
\psi(\bar{u},v)_{\mathcal{S}_2}&=&\arg \frac{ \vartheta(\bar{u}-v+\mathcal{A}_g,\Omega) \vartheta(u+v+\mathcal{B}_g,\Omega)}{\vartheta(\bar{u}+v+\mathcal{C}_g,\Omega) \vartheta(u-v+\mathcal{D}_g,\Omega)}=-\psi(u,v)_{\mathcal{S}_2} \nonumber 
\end{eqnarray}
if a $\mathcal{A}_g =\mathcal{D}_g+\null^tn_1I+\null^tm_1\Omega$  and $\mathcal{B}_g =\mathcal{C}_g+\null^tn_2I+\null^tm_2\Omega$ for some  $n_1, m_1, n_2, n_2 \in \mathbb{Z}^g$. The aim is to select $n_1, m_1, n_2, n_2$ (if it is possible)  to get $\psi(\bar{u},v)_{\mathcal{S}_2}=-\psi(u,v)_{\mathcal{S}_2}$ with no additional terms.  Using the translation properties (\ref{translat}) we can write

\begin{eqnarray}
\psi(\bar{u},v)_{\mathcal{S}_2}&=&-\psi(u,v)_{\mathcal{S}_2}+\arg e^{2\pi i( \null^tm_1(2v-\bar{u}-u-\mathcal{D}_g-\mathcal{A}_g)-\null^tm_2(2v+\bar{u}+u+\mathcal{C}_g+\mathcal{D}_g))}=\nonumber \\
&=& -\psi(u,v)_{\mathcal{S}_2} \nonumber 
\end{eqnarray}
if and only if $ m_1=m_2=0$. With these choices $\mathcal{A}_g =\mathcal{D}_g+\null^tn_1I$ and $\mathcal{B}_g =\mathcal{C}_g+\null^tn_2I$. We continue with:

\begin{eqnarray}
\psi(u,-\bar{v})_{\mathcal{S}_2}=\arg \frac{ \vartheta(\bar{u}-v+\bar{\mathcal{C}}_g,\Omega) \vartheta(u+v+\bar{\mathcal{D}}_g,\Omega)}{\vartheta(\bar{u}+v+\bar{\mathcal{A}}_g,\Omega) \vartheta(u-v+\bar{\mathcal{B}}_g,\Omega)}= 
-\psi(u,v)_{\mathcal{S}_2} \nonumber
\end{eqnarray}
if  $\mathcal{A}_g =\bar{\mathcal{B}}_g+\null^tn_3I+\null^tm_3\Omega$ and $\bar{\mathcal{C}}_g =\mathcal{D}_g+\null^tn_4I+\null^tm_4\Omega$ for some $n_3,m_3,n_4,m_4\in \mathbb{Z}^g$ to be determined.  We have used the property of first order  Riemann theta functions  $\overline{\vartheta(z,\Omega)}=\vartheta(-\bar{z},-\bar{\Omega})=\vartheta(\bar{z},\Omega)$ as the matrix $\Omega$ is pure imaginary. To get the reflection $\psi(u,-\bar{v})_{\mathcal{S}_2}=-\psi(u,v)_{\mathcal{S}_2}$ we write

\begin{eqnarray}
\psi(u,-\bar{v})_{\mathcal{S}_2} &=&-\psi(u,v)_{\mathcal{S}_2} +\arg e^{2\pi i (-\null^tm_3(\mathcal{A}_g+\mathcal{B}_g+2\Re eu)-\null^tm_4( \mathcal{D}_g+\mathcal{C}_g+2\Re e u )) }=\nonumber \\
&=&-\psi(u,v)_{\mathcal{S}_2} \nonumber
\end{eqnarray}
if and only if $m_3=-m_4$. This reflection imposes $\mathcal{A}_g =\bar{\mathcal{B}}_g+\null^tn_3I+\null^tm_3\Omega$, $\bar{\mathcal{C}}_g =\mathcal{D}_g+\null^tn_4I-\null^tm_3\Omega$ and $\mathcal{A}_g+\mathcal{B}_g-\mathcal{C}_g-\mathcal{D}_g \in i\mathbb{R}^g ~\mbox{mod}~\null^tm\Omega$. The reflections $\psi(u,v)_{\mathcal{S}_2}=-\psi(u,e^{(1)}+\dots+e^{(j)}-\bar{v})_{\mathcal{S}_2}$ give the same conditions on $\mathcal{A}_g, \dots, \mathcal{D}_g$ as translations respect to the columns of the  identity matrix induce no new relations for the odd characteristics. The last reflections we study involves the matrix of periods $\Omega$; by writing 
$\tau^{(1)}+\tau^{(j)}=\null^tn\Omega$, with $\null^tn=(1,0,\dots,0,\underbrace{1}_{j},0\dots,0)$ and $j=2,\dots,g$:

\begin{eqnarray}
&&\psi(\tau^{(1)}+\tau^{(j)}+\bar{u},v)_{\mathcal{S}_2}=-\psi(u,v)_{\mathcal{S}_2}+\arg e^{2\pi i \null^tm_1(2v-\bar{u}-u-\mathcal{D}_g-\mathcal{A}_g)}+\nonumber \\
&&+\arg e^{-2\pi i \null^tm_2(2v+\bar{u}+u+\mathcal{C}_g+\mathcal{D}_g) }+\arg e^{2\pi i \null^tn(4v+\mathcal{B}_g+\mathcal{C}_g-\mathcal{A}_g-\mathcal{D}_g) }, \nonumber
\end{eqnarray}
if  $\mathcal{A}_g =\mathcal{D}_g+\null^tn_1I+\null^tm_1\Omega$  and $\mathcal{B}_g =\mathcal{C}_g+\null^tn_2I+\null^tm_2\Omega$ for some  $n_1, m_1, n_2, n_2 \in \mathbb{Z}^g$. Thus we get $\psi(\tau^{(1)}+\tau^{(j)}+\bar{u},v)_{\mathcal{S}_2}=-\psi(u,v)_{\mathcal{S}_2}$ for $j=2,\dots,g$ if and only if 
$\mathcal{A}_g =\mathcal{D}_g+\null^tn_1I-\null^tn\Omega$  and $\mathcal{B}_g =\mathcal{C}_g+\null^tn_2I+\null^tn\Omega$ with $\null^tn=(1,1,\dots,,1)$. The reflection $\psi(\tau^{(1)}+\bar{u},v)_{\mathcal{S}_2}=-\psi(u,v)_{\mathcal{S}_2}$ fixes no new equation for $\mathcal{A}_g, \dots, \mathcal{D}_g$. In summary    $\psi(u,v)_{\mathcal{S}_2}=\arg \frac{ \vartheta(u-v+\mathcal{A}_g,\Omega) \vartheta(\bar{u}+v+\mathcal{B}_g,\Omega)}{\vartheta(u+v+\mathcal{C}_g,\Omega) \vartheta(\bar{u}-v+\mathcal{D}_g,\Omega)}$ satisfies 
 $\psi(u,v)_{\mathcal{S}_2}=-\psi(\bar{u},v)_{\mathcal{S}_2}=-\psi(u,-\bar{v})_{\mathcal{S}_2}=-\psi(\tau^{(1)}+\bar{u},v)_{\mathcal{S}_2}=-\psi(u,e^{(1)}-\bar{v})_{\mathcal{S}_2}=-\psi(\tau^{(1)}+\tau^{(j)}+\bar{u},v)_{\mathcal{S}_2}=  -\psi(u,e^{(1)}+\dots+e^{(j)}-\bar{v})_{\mathcal{S}_2}~~ \forall j=2,\dots,g $  if 
 
 \begin{align*}
 \mathcal{A}_g &=\mathcal{D}_g+\null^tn_1I & \mathcal{A}_g &=\bar{\mathcal{B}}_g+\null^tn_3I+\null^tm_3\Omega  & \mathcal{A}_g =\mathcal{D}_g+&\null^tn_1I-\null^tn\Omega \\
 \mathcal{B}_g &=\mathcal{C}_g+\null^tn_2I & \bar{\mathcal{C}}_g &=\mathcal{D}_g+\null^tn_4I-\null^tm_3\Omega  & \mathcal{B}_g =\mathcal{C}_g+&\null^tn_2I+\null^tn\Omega \\
 \end{align*}
 and $\mathcal{A}_g+\mathcal{B}_g-\mathcal{C}_g-\mathcal{D}_g \in i\mathbb{R}~\mbox{mod}~\null^tm\Omega$ with $n_1,n_2,n_3,m_3 \in \mathbb{Z}^g$ and $\null^tn=(1,1,\dots,,1)$.
 
 Of course the above conditions are not compatible. To solve this impasse we have  to  relax some of the reflections on the  map $\psi(u,v)_{\mathcal{S}_2}$. We decide to relax those respect to the columns of the matrix of periods: this approach will be compatible with the analysis of the zero modes contributions. Explicitly this amounts to get $\psi(u,v)_{\mathcal{S}_2}=-\psi(\tau^{(1)}+\tau^{(j)}+\bar{u},v)_{\mathcal{S}_2}+8\pi\re  v_1+8\pi\re  v_j$ for $\mathcal{A}_g =\mathcal{D}_g+\null^tn_1I$, $\mathcal{B}_g =\mathcal{C}_g+\null^tn_2I$ and $\mathcal{B}_g+\mathcal{C}_g-\mathcal{A}_g-\mathcal{D}_g \in i\mathbb{R}~\mbox{mod}~\null^tm\Omega$; comparing this new set of equations for  $\mathcal{A}_g, \dots, \mathcal{D}_g$ with those relative to the reflections respect to the identity matrix, $\bar{u}$ and $-\bar{v}$ we get the constraints for the odd half periods

\begin{align*}
\mathcal{A}_g &=\mathcal{D}_g & \mathcal{A}_g &=\bar{\mathcal{B}}_g &  &\mathcal{B}_g+\mathcal{C}_g-\mathcal{A}_g-\mathcal{D}_g \in i\mathbb{R}^g ~\mbox{mod}~\null^tm\Omega\\
\mathcal{B}_g &=\mathcal{C}_g &  \bar{\mathcal{C}}_g &=\mathcal{D}_g &   &\mathcal{A}_g+\mathcal{B}_g-\mathcal{C}_g-\mathcal{D}_g \in i\mathbb{R}^g  ~\mbox{mod}~\null^tm\Omega \\
 \end{align*}
with reflections for $\psi(u,v)_{\mathcal{S}_2}$ given by (\ref{ABCD}).

\end{proof}

\section{Relevant superpropagators: some properties}

In this appendix we show some properties of the relevant superpropagators $G_{\mathcal{S}_i}$ as gauge fixed homotopy operators and we give the proof of the Hodge-Kodaira splitting $dG_{\mathcal{S}_i}+G_{\mathcal{S}_i}d=I-P_{\mathcal{S}_i}$ appearing in Def.1. First of all the operators

\begin{equation}\label{operato}
(G_{\mathcal{S}_i}\phi)(Q):= \int_{P_n} \theta(Q,P)_{\mathcal{S}_i} \wedge \phi(P) ~~~~Q \in P_n, 
\end{equation}
present an integrable singularity along the diagonal in $P_n \times P_n$. We continue with

\begin{Prop}
Let $G_{\mathcal{S}_i}$ be given by $(\ref{operato})$;  they realize the Hodge-Kodaira splitting 

\begin{equation}
dG_{\mathcal{S}_i}+G_{\mathcal{S}_i}d=I-P_{\mathcal{S}_i} \label{spl}
\end{equation} 
of the differential complexes $(\mathcal{H}^n_{\mathcal{S}_i},d)$, where $P_{\mathcal{S}_i}$ denotes  projection onto cohomology.
\end{Prop}

\begin{proof}
Let $\phi \in \mathcal{H}^{n,0}_{\mathcal{S}_i}$;  the splitting equation (\ref{spl}) reduces to $(G_{\mathcal{S}_i}d\phi)(v)=\phi(v)$, where, as usual, $(u,v)$ denotes the coordinates of points $(P,Q)$ of the polygon  $P_n$. Explicitly we get 
\begin{eqnarray*}
&&(G_{\mathcal{S}_i}d\phi)(v)=\int_{P_n-B_{\mu}(v)}\frac{1}{2\pi}d_u\mu_{\mathcal{S}_i}(u,v)\wedge d\phi(u)+\nonumber \\
&&-\left\{\begin{array}{c} \int_{P_n}4\im  v_i(\im \Omega)_{ij}^{-1} d\re  u_j \wedge d\phi(u)  ~~\mbox{\tiny{i=1}} \\ \int_{P_n}4\im  u_i(\im \Omega)_{ij}^{-1} d\re  v_j \wedge d\phi(u)~~\mbox{\tiny{i=2}} \end{array} \right.=\phi(v)
\end{eqnarray*}
as in the $i=1$ case integrating by parts we get the sum  $\int_{\partial P_n}\frac{1}{2\pi}d_u\mu_{\mathcal{S}_i}(u,v)\wedge \phi(u)-4\int_{\partial P_n }\im  v_i(\im \Omega)_{ij}^{-1} d\re  u_j \wedge \phi(u)$ which  vanishes for the boundary conditions of the integral kernels and $\phi$ once we add the null term $\int_{\partial P_n}\frac{1}{2\pi}d_v\mu_{\mathcal{S}_i}(u,v)\wedge \phi(u)$. In the $i=2$ case, we can add the terms $4\int_{\partial P_n}\im u_i (\im  \Omega)^{-1}_{ij}d\re  v_j\wedge \phi(u)$ and  $\int_{\partial P_n}\frac{1}{2\pi}d_v\mu_{\mathcal{S}_i}(u,v)\wedge \phi(u)$   repeating the same trick of the $i=1$ case to eliminate the contributions over $\partial P_n$. The orientation of $\partial P_n$ is taken to be counterclockwise.
Let $\phi \in \mathcal{H}^{n,1}_{\mathcal{S}_i}$ now; as the cohomology of the differential complexes in not trivial in degree 1, then we have to distinguish different cases. We begin with  $\phi=f(u,\bar{u})du+g(u,\bar{u})d\bar{u}$; the splitting equation and Lemma 3 give 
$(dG_{\mathcal{S}_i}\phi)(v)+(G_{\mathcal{S}_i}d\phi)(v)= d\bar{v}\int_{\partial B_{\mu}(v)}\frac{1}{4\pi i}(\frac{g}{u-v}+\frac{f}{\bar{u}-\bar{v}})du+  \nonumber \\
-dv\int_{\partial B_{\mu}(v)}\frac{1}{4\pi i}(\frac{g}{u-v}
+\frac{f}{\bar{u}-\bar{v}})d\bar{u}+\int_{\partial P_n}\frac{1}{2\pi}d_v\mu_{\mathcal{S}_i}(u,v)\wedge\phi(u)-\int_{\partial B_{\mu}(v)}  \frac{1}{2\pi}d_v\mu_{\mathcal{S}_i}(u,v)\wedge\phi(u)-\left\{\begin{array}{c} 4\int_{P_n}d\im v_i (\im  \Omega)^{-1}_{ij}d\re  u_j\wedge\phi(u)~~i=1 \\ 4\int_{P_n}\im u_i (\im  \Omega)^{-1}_{ij}d\re  v_j\wedge d\phi(u)~i=2  \end{array}\right.$. All the line integrals are taken counterclockwise. The sum of the integrals over the boundary $\partial B_{\mu}(v)$ give precisely the identity on the right hand side of the splitting equation; in both the $i=1,2$ cases we eliminate the terms along $\partial P_n$ by adding suitable terms as in the $\phi \in \mathcal{H}^{n,0}_{\mathcal{S}_i}$ case; in particular for $i=2$ we get the projection onto cohomology integrating by parts.

Selecting an exact one form, $\phi=df$ with $f \in \mathcal{H}^{n,0}_{\mathcal{S}_i}$  we get no projection as in the $i=1$ case $4\int_{P_n}d\im v_i (\im  \Omega)^{-1}_{ij}d\re  u_j\wedge df(u)=-4\int_{\partial P_n }d\im v_i (\im  \Omega)^{-1}_{ij}d\re  u_j\wedge f(u)$ which, summed to $\int_{\partial P_n}\frac{1}{2\pi}d\mu_{\mathcal{S}_2}(u,v)\wedge\phi(u)$, gives as usual zero (the $i=2$ case is immediate). If $\phi$ is a linear combination of zero modes we have the following

\begin{Lem}
Let  $\phi \in H^1( \mathcal{H}^{n}_{\mathcal{S}_i})$; then 

\begin{eqnarray}
(P_{\mathcal{S}_i}\phi)(Q)=\phi(Q).
\label{projeee}
\end{eqnarray}

\end{Lem}

\begin{proof}

Let $(P_{\mathcal{S}_i}\phi)(Q)= \left\{ \begin{array}{c} 4\int_{P_n}d\im \varphi_i(Q) (\im  \Omega)^{-1}_{ij}d\re  \varphi_j(P)\wedge \phi(P), \\  4\int_{P_n}d\im \varphi_i(P) (\im  \Omega)^{-1}_{ij}d\re  \varphi_j(Q)\wedge \phi(P) \end{array} \right.$ with $\varphi_i(\cdot)$ the $i^{th}$ component of the Abel map $\varphi$ on $\mathcal{M}$; we prove (\ref{projeee}) for  $i=1$: the other case is analog. We want to use Riemann's bilinear relations for holomorphic differentials on the hyperelliptic curve $\mathcal{M}$ of Section 6.1. Let $\mathcal{M}$ be given as in Figure 1; the hyperelliptic involution $J$ is seen as a rotation by $\pi$ radians about an axis passing through the $2g+2$ branching points. Let $\gamma_j$ be an oriented curve from $P_{2j-1}$ to $P_{2j}$ for $j=1,\dots,g$. We have defined the canonical homology basis $\mathcal{B}_1=\{a_1,\dots,a_g,b_1,\dots,b_g\}$ where  the curve $a_j$ is $\gamma_j$ followed by $-J\gamma_j$, i.e. $a_j$  joins $P_{2j-1}$ to $P_{2j}$ and returns to $P_{2j-1}$ and the curve $b_j$ joins a point on $a_{j}$ to a point to $\alpha$, returning to the point on $a_j$.
  
  Let  $\beta_j$ be the curve that joins $P_{2j+1}$ to $P_{2j}$ and returns to $P_{2j+1}$ for $j=1,2\dots,g$ and $\alpha$ the one which joins $P_{2g+2}$ to $P_{2g+1}$ and returns to $P_{2g+2}$ (see Figure 3). It follows that $\alpha \cdot a_1=\alpha \cdot a_2 = \dots = \alpha \cdot a_g=0$, $\alpha \cdot b_k=1$ $\forall k=1,2\dots,g$ so $\alpha = a_1+a_2+ \dots a_g$ in homology, $\beta_i \cdot b_j =\beta_i \cdot a_j=0$ for $ i \not= j, \beta_g \cdot \alpha=1$, $\beta_i \cdot a_{i+1} =-1$ $\forall i=1,\dots,g-1$, $\beta_i \cdot a_{i} =1 $ $\forall i=1,\dots,g$. Thus up to homology we conclude $\beta_i=b_{i+1}-b_i$ $\forall i=1,\dots,g-1$ and $\beta_g=-b_g$. In the following $\{ \omega_1, \dots ,\omega_g \} \equiv \{\underline{\omega}\}$ is the  basis of holomorphic differentials  dual to the canonical homology basis $\mathcal{B}_1$.

  \vspace{0.5cm}
   \begin{figure}[h!]
   \begin{center}
   \epsfig{figure=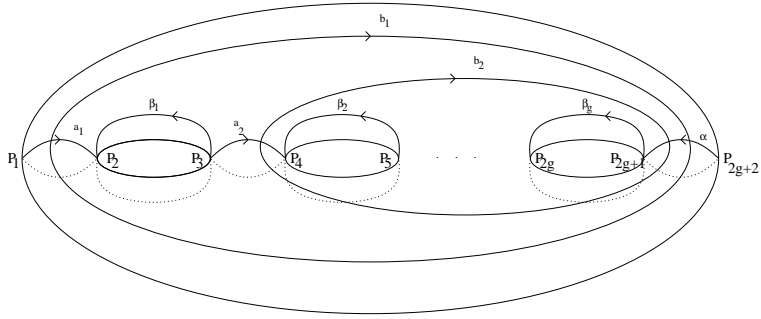,width=10cm}
 \caption{Homology basis $\mathcal{B}_1$ and $\{\beta_i,\alpha \}_{i=1,\dots,g}$ curves}   
\end{center}
   \end{figure}

Writing $\phi(P)=\alpha_kd\im \varphi_k(P)$, with $\{d\im \varphi_k(P) \}_{k=1,\dots,g}$ basis of $H^1(\mathcal{H}^{n}_{\mathcal{S}_1})$ then 
 
\begin{eqnarray}
(P_{\mathcal{S}_1}\phi)(Q)&=&-i\alpha_k d\im  \varphi_i(Q)(\im \Omega)^{-1}_{ij}(\int_{P_n}\omega_j(P)\wedge\overline{\omega_k(P)}+\nonumber \\
&&+\int_{P_n}\omega_k(P)\wedge\overline{\omega_j(P)}). \label{decomp}
\end{eqnarray}
The $\underline{\omega}(P)$ are closed differentials;  considering the fundamental polygon $\Gamma$ with symbol $\prod_{i=1}^{g}a_jb_ja^{-1}_jb^{-1}_j$  associated to the curve $\mathcal{M}$ we can write $\int_{P_n}\omega_j(P)\wedge\overline{\omega_k(P)}=\int_{\partial P_n}f_j(P)\wedge\overline{\omega_k(P)}$ where $df_j(P)=\omega_j(P)$. Then $\int_{P_n}\omega_j(P)\wedge\overline{\omega_k(P)}=\frac{1}{4}\sum_{i=1}^g[\int_{a_i}f_j(P)\wedge\overline{\omega_k(P)}+\int_{a^{-1}_i}f_j(P)\wedge\overline{\omega_k(P)}+\int_{b_i}f_j(P)\wedge\overline{\omega_k(P)}+\int_{b^{-1}_i}f_j(P)\wedge\overline{\omega_k(P)}]=\frac{1}{2}\tau_{jk}$; repeating the same procedure for the second summand of the r.h.s in (\ref{decomp}) we get $(P_{\mathcal{S}_1}\phi)(Q)=-i\alpha_k d\im  \varphi_i(Q)(\im \Omega)^{-1}_{ij}i(\im \Omega)_{jk}=\alpha_kd\im  \varphi_{k}(Q)$.

\end{proof}

 The above Lemma  concludes the proof of (\ref{spl}) for the $\phi \in \mathcal{H}^{n,1}_{\mathcal{S}_i}$ case. 
 
 If $\phi \in \mathcal{H}^{n,2}_{\mathcal{S}_i}$, then $\phi=d\omega$ for $\omega\in \mathcal{H}^{n,1}_{\mathcal{S}_i}$ (Lemma 1). This implies that we can repeat essentially the calculations of the preceding case: the splitting equation reduces to 
$(dG_{\mathcal{S}_i}\phi)(v)= \phi(v)$ as we have no projection in degree 2;  this concludes the proof of Proposition 5.

\end{proof}

\end{document}